\documentclass[12pt]{article}
\pdfoutput=1 
\usepackage{setspace}
\usepackage[T1]{fontenc}
\usepackage{lmodern}
\usepackage[colorlinks=true, pdfstartview=FitV, linkcolor=blue,citecolor=blue, urlcolor=blue]{hyperref}

\usepackage{microtype}
\usepackage[margin=1in]{geometry}
\usepackage{setspace}
\usepackage[title]{appendix}
\usepackage[round]{natbib}
\usepackage[dvipsnames]{xcolor}
\usepackage{amsmath,amsthm,amssymb,ifthen,mathrsfs,amsfonts,bm}
\usepackage{empheq}
\usepackage{mathtools}
\usepackage{mathabx}
\usepackage{subcaption}
\usepackage{graphicx}
\usepackage{accents}
\usepackage{enumitem}
\usepackage{bigints}
\usepackage{bookmark}
\usepackage{chngcntr}

\usepackage[framemethod=tikz]{mdframed}
\usepackage{pgf,tikz}
\usetikzlibrary{calc,math}
\usetikzlibrary{arrows,shapes,positioning,automata}
\usetikzlibrary{decorations.markings}
\usetikzlibrary{shapes.geometric}
\usepackage{pgfplots}
\usetikzlibrary{intersections}
\pgfplotsset{compat=newest}
\usepgfplotslibrary{fillbetween}
\usetikzlibrary{patterns}
\newtheorem{theorem}{Theorem}
\newtheorem{lemma}{Lemma}
\newtheorem{proposition}{Proposition}
\theoremstyle{definition}
\newtheorem{corollary}{Corollary}
\newtheorem{definition}{Definition}

\newtheorem{example}{Example}

\newtheorem{remark}{Remark}

\newcommand{\BE}{\mathsf{E}}
\newcommand{\BP}{\mathsf{P}}
\newcommand{\BR}{\mathbb{R}}

\renewcommand{\d}{\mathrm{d}}
\newcommand{\ve}{\varepsilon}

\numberwithin{equation}{section}

\makeatletter
\newcommand{\rn}[1]{\romannumeral #1}
\newcommand{\Rn}[1]{\expandafter\@slowromancap\romannumeral #1@}
\makeatother

\title{Reputation, Learning and Project Choice in Frictional Economies}
\date{May 2024}
\author{Farzad Pourbabaee
\thanks{I am grateful to my advisors Bob Anderson, Gustavo Manso and Chris Shannon for the guidance and support. For helpful conversations and comments, I thank Haluk Ergin, Peter B. McCrory, Joseph Root, Ali Shourideh, David Sraer, Philipp Strack, and Pierre-Olivier Weill. I received valuable feedback from the referees as well as seminar participants at Search and Matching in Macro and Finance Workshop and UC Berkeley.}
\thanks{California Institute of Technology. \href{mailto:far@caltech.edu}{far@caltech.edu}}}
\renewcommand\footnotemark{}

\begin{document}
\maketitle
\begin{abstract}
I introduce a dynamic model of learning and random meetings between a long-lived agent with unknown ability and heterogeneous projects with observable qualities. The outcomes of the agent's matches with the projects determine her posterior belief about her ability (i.e., her reputation). In a self-type learning framework with endogenous outside option, I find the optimal project selection strategy of the agent, that determines what types of projects the agent with a certain level of reputation will accept. Sections of the optimal matching set become \textit{increasing intervals}, with different cutoffs across different types of the projects. Increasing the meeting rate has asymmetric effects on the sections of the matching sets: it unambiguously expands the section for the high type projects, while on some regions, it initially expands and then shrinks the section of the low type projects.
\end{abstract}

\clearpage
%%%%%%%%%%%%%%%%%%%%%%%%%%%%%%%%%%%%%%%%%%%%%
\setstretch{1.35}
\tableofcontents
\clearpage

%%%%%%%%%%%%%%%%%%%%%%%%%%%%%%%%%%%%%%%%%%%%%%%%%%%%%%%
\section{Introduction}
Most of the theoretical literature on experimentation and project choice revolves around learning the other party's (namely the project's) type. In this paper, however, I turn the focus to learning the self-type, and thus engaging in a \textit{self-experimentation} setting. Specifically, the agent in my paper does not know her type, and the only way to learn it is by matching with the projects and observing their outcomes.

There are natural instances where agents \textit{learn} their type through the course of their matches with other parties. For example, firms learn about their productivity while they are matched with workers. Venture capitalists learn about their ability and the quality of their post-investment services while investing in the startups. Common in all these cases is the cost of maintaining the match and the \textit{tangible} created surplus (such as the high-quality output of production in the first case, and the startups' success in the second case). These tangible gains can be isomorphically captured by the choice of the  matching function, that takes in the types of partners and returns the output.

However, when the agent holds incomplete information about her type, there is also an \textit{intangible} gain due to the learning, that cannot be nested in the former construct. Because, what is now used as an input to the matching function is no longer the static type of the agent, but a dynamic state process that reflects the agent's belief of her own type. Specifically, in addition to the tangible gains, there are now information gains from agent's project choices, as present selections convey information about the agent's ability, that in turn can be used in future choices of projects. The basic research question that I pose and address in this paper is as follows: When confronted with diverse projects that vary in their expected payoffs, what constitutes the agent's optimal project selection policy in relation to her reputation?

In this economy, the agent is ex ante endowed with a high or low \textit{immutable} type $\theta \in \{L,H\}$, that is hidden to herself. On the other hand, there are heterogeneous projects with observable qualities denoted by $q \in \{a,b\}$, which I often refer to as $a$-projects and $b$-projects. The agent randomly meets the projects subject to the search frictions and decide whether to accept them or not. Once a project is accepted, there will be a random success event whose arrival intensity depends on the types of the agent and the project (namely on $\theta$ and $q$). The agent continuously updates her belief about the underlying type during the course of her matches. Therefore, I interpret the posterior belief as her reputation and denote it by $\pi$.

Whenever the agent pairs up with a project, a learning opportunity is created about her type. Since maintaining the match is costly, the agent effectively solves a stopping time problem, in which she weighs the \textit{matching value function} $v_q(\pi)$ (that is a function of her current reputation $\pi$ and the type of the project $q$) against the reservation value $w(\pi)$ --- the value of holding the current reputation while being unmatched, that is called the \textit{reputation value function} throughout the paper. Because of the random meetings framework, these two functions are intertwined in the fixed-point. That is the reputation function is simply the expected discounted value of future surpluses that the agent extracts, and the matching value function is the solution to the free boundary problem with the exit option of $w$. The continuation region of this free boundary problem determines the optimal matching set $\mathcal{M}$, that in turn defines the acceptable levels of reputation with which the agent selects and holds on to a particular project. Specifically, $(q,\pi) \in \mathcal{M}$ if $v_q(\pi) > w(\pi)$. In light of this specification and following the terminology of the optimal stopping literature, I use the \textit{matching set} and \textit{continuation region} interchangeably and both refer to the subset $\mathcal{M}$.

The central innovation of this paper is to study the optimality and shape of these matching sets (namely the continuation regions) when the agent has long-run incentives and learn her ability as she selects and matches with the projects. Specifically, I find and study the properties of the optimal tuple $\langle w^*,v^*,\mathcal{M}^*\rangle$. The main point of the departure from the experimentation literature (e.g., \cite{keller2005strategic}) is the endogeneity of the outside option $w$, that determines the types of acceptable projects in the agent's optimal policy. In addition, the subject of learning in the experimentation literature is the project's type, whereas in my paper the learning is about the self-type and the projects provide the context for learning and a source of creating surplus.
%-----------------------------------------

\subsection{Organization of Results}
In Section~\ref{sec: model}, I introduce the dynamic learning and project selection model. Three main objects in the study of agent's optimal policy are the matching value function $v_q(\pi)$, the reputation or reservation value function $w(\pi)$, and the matching set $\mathcal{M}$. We will see how the agent's optimal policy can be translated into a fixed-point solution of a system that connects the three elements mentioned above.

In Section~\ref{sec: Opt_with_learning}, I find the unique optimal tuple $\langle w^*,v^*,\mathcal{M}^*\rangle$ in the space of continuously differentiable value functions, i.e., $C^1[0,1]$. I study the properties of the value functions (such as monotonicity and convexity), and show the sections of the optimal matching set are \textit{increasing intervals}, and hence the agent's optimal policy is to stay matched with a project so long as her reputation is larger than a certain threshold. The threshold for high type projects is lower, and hence the agent shows more tolerance for failure when matched to the high types. In particular, by letting $\mathcal{M}^*_q$ represent the interval of reputation levels where the agent remains matched with a $q$-project, it follows that at the optimum $\mathcal{M}^*_a \subseteq \mathcal{M}^*_b$. Due to the search frictions, there is a cost region that $\mathcal{M}^*_a \neq \emptyset$, and thus even the low type projects get selected.

I present the qualitative features of the matching sets and the value functions in Section~\ref{sec: Qual}. In particular, to uncover the unique role of learning on the shape of the optimal matching sets and value functions, I study the \textit{no-learning} counterpart of the original model in Section~\ref{subs: no_learning_version}. Specifically, I assume that the agent's true type is equal to her reputation (i.e., a number $\pi \in [0,1]$), as opposed to a background hidden binary variable $\theta\in \{L,H\}$. This will shut down the learning channel, that is the Bayesian learning force will be absent in the associated Bellman equations. 

Subsequently, in this no-learning environment, I find the unique optimal outcome $\langle \hat w, \hat v, \widehat{\mathcal{M}} \rangle$ in the space of continuous functions $C[0,1]$. The matching value function $\hat v$ is locally concave with kinks on the boundaries. Namely, it is no longer convex and continuously differentiable despite its counterpart $v^*$ in the original learning model. Losing convexity in the value functions (due to the absence of learning incentives) leaves the matching sets smaller than their learning counterpart $\mathcal{M}^*$. Lastly, I show in this setting that lowering the search frictions \textit{symmetrically} expands both sections of the matching set (i.e., $\widehat{\mathcal{M}}_a$ and $\widehat{\mathcal{M}}_b$).

Next, in Section~\ref{subs: comp_statics}, I present the comparative statics of the original learning model. Specifically, based on the unique existence of the optimal tuple $\langle w^*, v^*, \mathcal{M}^*\rangle$ in the space of $C^1$ functions, I present the comparative statics of this tuple with respect to the primitives of the economy. An important one among them is the impact of the search frictions on the size of the sections of the optimal matching set, namely $\mathcal{M}^*_a$ and $\mathcal{M}^*_b$. I show that decreasing the search frictions unambiguously expands the high type section $\mathcal{M}^*_b$, but on some regions, initially expands and then shrinks the low type section $\mathcal{M}^*_a$. This \textit{asymmetric} response to search frictions was not present in the no-learning matching sets $(\widehat{\mathcal{M}}_a, \widehat{\mathcal{M}}_b)$.

Building upon the baseline results with a single agent, Section~\ref{sec: rep_exter} examines an economy characterized by \textit{reputational externality} and populated by a continuum of agents. Specifically, by employing a reputation weight function, the meeting rate of each agent becomes higher as her reputation increases. More importantly, this rate is inversely influenced by the steady-state distribution of reputation weights across the population. Hence the more reputable agents slow down the meeting rate of the less reputable ones. As a result of this externality, I demonstrate that a marginal reduction in the symmetric equilibrium termination point increases the social surplus. Consequently, the equilibrium social surplus suffers from agents under-learning their self-types and early termination of the projects.

Lastly, the paper concludes in Section~\ref{sec: conclusion}.
%------------------------------------------------
\subsection{Related Literature}
The Bayesian learning force in the agent's decision problem in this paper is based on the exponential arrival of breakthroughs, and in that sense the paper is related to the experimentation literature with exponential news processes, initiated by \cite{keller2005strategic}, and expanded in the follow-up works of \cite{keller2010strategic} and~\cite{keller2015breakdowns}. The exponential Bandit approach has also been applied to other strategic settings with payoff and informational externalities between players \citep{margaria2020learning,das2020strategic,boyarchenko2020super,boyarchenko2021inefficiency}. The common theme in this line of research revolves around the uncertainty of the project's type for the decision maker(s). In contrast, in the present paper, the projects' types are observable and they provide a context for the decision maker to learn her own type while they are being selected.

In the context of reputation building (when the information about the persistent or dynamic self-type is incomplete) and interpreting the reputation as the posterior belief, this paper is related to \cite{holmstrom1999managerial, board2013reputation, bonatti2017career}. However the kind of economic engagement that releases informative signals in these papers is the agent's effort, and in the current study is about the agent's project selection.

The analysis of this paper has also the flavor of the literature on learning within the labor markets such as the works by \cite{jovanovic1979job,moscarini2005job, li2017efficient}. Aside from differences in the context and motivation, the subject of learning in these studies is the \textit{match-specific} parameter, and not the underlying types of the agents. Therefore, the information released over the present match has no bearing on the future matches and naturally the reputational aspects are absent.

There is also previous research on how agents hold \textit{perfect} private information about themselves, and receive some form of information about the type of their partner before the match \citep{chakraborty2010two,liu2014stable}. My setting is different from these works, mainly in the sense that the agent in this paper has incomplete information about herself, and one of the motives in her matching decisions (besides receiving the tangible surplus from the projects she accepts) is learning her type.

My findings also contribute to the literature about Bandits with correlated arms \citep{camargo2007good, rosenberg2013games}. This problem is known to be difficult, and thus very little has been achieved in economics literature. The self-experimentation model that is developed in this paper is formally equivalent to a two-armed Bandit setting, where the arms' payoffs are correlated. Specifically, what correlates the payoffs of the projects is the single dimensional variable that represents the agent's hidden type.

%%%%%%%%%%%%%%%%%%%%%%%%%%%%%%%%%%%%%%%%%%%%%%%%%%%%%%%
\section{Model}
\label{sec: model}
%----------------------------------------------------------------------------
\subsection{Agent, Projects and Dynamic Timeline}
In this part, I describe the elements of an economy populated by a single long-lived agent, who discounts future payoffs, and a continuum of projects with observable types.

\paragraph{Agent.} The agent is a long-lived individual with the rate of time preference $r>0$. She holds incomplete information about her immutable type $\theta\in \{L,H\}$. The $\sigma$-field $\mathcal{I}_t$ aggregates all the information that is available in the economy at time $t\in \BR_+$. The agent cares about her reputation, which is the posterior belief about her type. Given the filtration $\bm{I}=\left\{\mathcal{I}_t\right\}$, $\pi_t = \BP\left(\left. \theta=H\right| \mathcal{I}_t\right)$ reflects her time-$t$ reputation. 

\paragraph{Projects.} The entities on the other side of this economy are treated as projects that are selected by the agent. Specifically, they have no bargaining power against the agent.\footnote{This assumption makes the analysis substantially simpler, yet it downplays the strategic role of ``project owners'' in the optimal outcome. However, given the paper's focus on the agent's side and her reputational concerns, such an abstraction seems plausible.} Each project is endowed with a type $q \in \{a,b\}$, which is publicly observable. The (unnormalized) mass of type-$q$ projects is $\varphi_q$ for $q \in \{a,b\}$, exogenously replenished and held constant.

\paragraph{Meetings and project selection.} The agent randomly meets the projects subject to the search frictions, with the meeting rate of $\kappa>0$, and the matching technology is \textit{quadratic}. That is the probability with which the agent meets a type-$q$ project over an infinitesimal period $\d t$ is approximately equal to $\kappa \varphi_q\,\d t$. Furthermore, the matches are one-to-one, that is both parties have capacity constraint over the number of partners they can accept.

\paragraph{Output and reputation.} Given the selection of a type-$q$ project by the type-$\theta$ agent, the success arrives at the rate of $\lambda_q$, if $\theta = H$. Otherwise, there will be no success. This means breakthroughs are \textit{conclusive} about the agent's ability.\footnote{The analysis involving inconclusive breakthroughs, where success can occur even with a low-type agent, presents several intractable steps and is therefore excluded. Due to its tractability, conclusive breakthroughs are used in a number of recent studies \citep{bonatti2017career, margaria2020learning, das2020strategic, boyarchenko2021inefficiency}. Additionally, choosing exponential processes to model breakthroughs is more suitable when news arrives at discrete and randomly spaced intervals, in contrast to the Wiener process approach for experimentation \citep[e.g., see][]{bolton1999strategic, pourbabaee2020robust}.} Type-$b$ projects are superior to type-$a$, in the sense that $\lambda_b>\lambda_a$. The agent has to cover the flow cost of $c>0$ that is common across all projects --- and, for a non-trivial setting, one has to assume $c<\lambda_a<\lambda_b$. In return, she receives the right to terminate the project at her will, so conceptually a stopping time problem is solved by the agent ex post to every selection of a project. The flow cost $c$ captures both the running cost of the project and learning about the self-type $\theta$. I assume there is a mechanism in the economy that tracks the output of each project and records the Bayes-updated posterior of the agent. This information is reflected in the filtration $\bm{I}=\{\mathcal{I}_t\}$. The posterior dynamics for the reputation process (resulted from the Bayes law) follows
\begin{equation}
\label{eq: Poisson_Bayes}
	\d \pi_t = \left\{\begin{array}{cc}
	    -\pi_t(1-\pi_t)\lambda_q\,\d t &  \text{ before the breakthrough}\,, \\ 
	    0 & \text{ after the breakthrough}\,. 
	\end{array}\right.
\end{equation} 
Since the breakthroughs are conclusive, when the success arrives, $\pi_t$ promptly jumps up to one and remains there indefinitely.

Figure~\ref{fig: Dynamic_timeline} illustrates the dynamic timeline for the agent, who starts the cycle with reputation $\pi$, and after some exponentially distributed time meets a project randomly drawn from the population of available ones. A decision to accept or reject the contacting project is made by the agent. Upon rejection, she returns to the initial node, and conditioned on acceptance an investment problem with the flow cost of $c$ is solved. 

Finally, I interpret success as the occurrence of a breakthrough before the agent stops the project, leading her to rationally updating her belief upwards. And the failure refers to the outcome where the project is terminated before a success arrives, thus the agent returns to the unmatched status with a lower reputation. Importantly, after a success or a failure, the match is dissolved and both the agent and the project become available.

The self-experimentation model presented here is formally equivalent to a special two-armed Bandit setting, in which the payoffs to the two arms are correlated. Specifically, the agent's type $\theta$ is the variable that correlates the payoffs of the two arms. Therefore, the upcoming analysis can also be read through the lens of the experimentation literature with correlated arms.
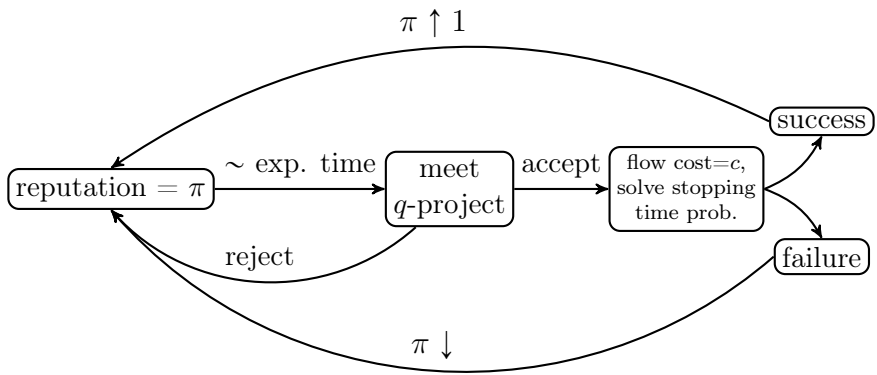
\begin{figure}[htp]
\begin{center}
\begin{tikzpicture}[scale=0.9,thick,>=stealth',dot/.style = {draw,fill = black, circle,inner sep = 0pt,minimum size = 4pt}]
		\node[draw,thick,rounded corners,inner sep=0.1cm,font=\small] (belief) at (0,0) {reputation = $\pi$};
		
		\node[draw,thick,rounded corners,inner sep=0.1cm,font=\small,align=center] (investee) at (5,0) {meet \\$q$-project};
				
		\draw [->] (belief) -- 	node [above,align=center,midway,font=\small]{$\sim$ exp. time} (investee);
		
		\path (investee) edge [->,bend left=45]  node[above,midway,font=\small] {reject} (belief.south);
		
		\node[draw,thick,rounded corners,inner sep=0.1cm,font=\scriptsize,align=center] (cost) at (8.5,0) {flow cost=$c$,\\solve stopping\\ time prob.};
		
		\path (investee) edge [->]  node[above,midway,font=\small] {accept} (cost);
		
		\node[draw,thick,rounded corners,inner sep=0.1cm,font=\small,align=center] (success) at (10.5,1) {success};

		\node[draw,thick,rounded corners,inner sep=0.1cm,font=\small,align=center] (failure) at (10.5,-1) {failure};
		
		\path (cost.east) edge [->,bend right=20]  (success.south);
		
		\path (cost.east) edge [->,bend right=-20]  (failure.north);
		
		\path (success.west) edge [->,bend left=330]  node[above,midway] {$\pi \uparrow 1$} (belief.north);

		\path (failure.west) edge [->,bend left=45]  node[above,midway] {$\pi \downarrow$} (belief.south);
\end{tikzpicture}
\end{center}
\caption{Decision timeline for the agent}
\label{fig: Dynamic_timeline}
\end{figure}
%%%%%%%%%%%%%%%%%%%%%%%%%%%%%%%%%%%%%%%%%%%%%%%%%%%%%%%%%%%%%%%%%%%%%%%%%%%%%%%%%%%%%

\subsection{Value Functions and Matching Sets}
In this section, I show that the agent's optimal strategy can be encapsulated by the choice of the matching sets. Additionally, I present the \textit{necessary} Bellman equations that every optimal $C^1$ value function, associated with the optimal matching sets, must satisfy. 

Let $w(\pi)$ be the optimal value of holding reputation $\pi$, when the agent is \textit{unmatched}. This function shall be treated as the agent's outside option and is weighed against the optimal matching value function upon the meetings.\footnote{Henceforth, I often drop the word ``optimal'', as it is clear from the context.} The matching value function when the agent with reputation $\pi$ selects and stays with a type-$q$ project is $v_q(\pi)$, that is the expected value of discounted future payoffs generated by this project. 

The optimality of the match between the agent of reputation $\pi$ and a type-$q$ project requires that $v_q(\pi) > w(\pi)$, in that case I say $(q,\pi) \in \mathcal{M} \subseteq \{a,b\}\times [0,1]$, where $\mathcal{M}$ is called the matching set (or interchangeably the continuation region). Also, understood from the context, $\mathcal{M}(\pi)$ (respectively, $\mathcal{M}_q$) refers to the $\pi$ (respectively, $q$) \textit{section} of this two dimensional set. In addition, often in the paper I use the indicator function $\chi_q(\pi)$ to denote whether the agent with reputation $\pi$ accepts a type-$q$ project, that is whether $(q,\pi)\in \mathcal{M}$ or not. 

\newpage
Recall that $\varphi$ represents the mass of available projects in the economy, which is treated exogenously as the primitives of the model. The agent also meets type-$q$ projects at the rate of $\kappa \varphi_q$. If a $q$-project is acceptable, it leads to a surplus of $v_q(\pi)-w(\pi)$ for the agent. Hence, the following Bellman equation falls out:\footnote{The reader can refer to Appendix~\ref{subs: Bellman_derivation} for a heuristic derivation of the Bellman equations~\eqref{eq: reservation_value} and~\eqref{eq: Bellman}.}
\begin{equation}
\label{eq: reservation_value}
	rw(\pi)= \kappa \sum_{q\in \mathcal{M}(\pi)} \varphi_q \big(v_q(\pi)-w(\pi)\big) .
\end{equation}

Next, I formally define the optimal matching value function, $v_q(\pi)$, and present the necessary Bellman equation that it satisfies.
 
Imagine a match between the agent with an initial reputation of $\pi$ and a type-$q$ project. Let $\sigma_q$ represent the random exponential time of success with the unit payoff and the arrival intensity of $\lambda_q$ if $\theta=H$. Therefore, the matching value function $v_q(\cdot)$ is an endogenous outcome of a free boundary problem with the outside option of $w(\cdot)$. In that, the agent selects an optimal stopping time $\tau$, upon which she stops backing the project, taking into account the project's success payoff and her reputation value $w$:\footnote{Henceforth, for a random variable $X$ and an event $A$, I use the convention that $\BE[X;A] := \BE[X \mathbf{1}_A]$, where $\mathbf{1}_A$ is the indicator function of the event $A$.}
\begin{equation}
\label{eq: matching_value_function}
	\begin{gathered}
		v_q(\pi) = \sup_\tau \left\{\BE\left[e^{-r\sigma_q}-c\int_0^{\sigma_q} e^{-rs}\d s+e^{-r \sigma_q} w(\pi_{\sigma_q}); \sigma_q \leq \tau\right]\right.\\
		+\left.\BE\left[-c\int_0^\tau e^{-r s}\d s+e^{-r \tau}w(\pi_\tau);\sigma_q>\tau\right]\right\}.
	\end{gathered}	
\end{equation}
Specifically, $\tau$ is adapted to the filtration $\bm{I}$. Namely, $\{\tau\leq t\}$ is $\{\pi_s: s<t\}$-measurable for every $t\in \BR_+$.

Formally, in the above stopping time problem, if the success happens before the agent stops backing the project (namely when $\sigma_q \leq \tau$), the agent collects the discounted unit payoff, has paid the flow cost until time $\sigma_q$, and successfully leaves the project with the updated reputation function $w(\pi_{\sigma_q})$. Observe that $\pi_{\sigma_q}$ is the updated posterior belief reflecting the successful exit, hence $\pi_{\sigma_q}=1$.

On the other hand, if the agent stops the project before the realization of success (specifically when $\tau<\sigma_q$), she only incurs the flow cost up to time $\tau$, and leaves with the updated reputation function $w(\pi_\tau)$ --- reflecting the fact that the success has not happened until time $\tau$. Therefore, the exit option at the stopping time $\tau$ is the agent's reservation value of holding reputation $\pi_\tau$, i.e., $w(\pi_\tau)$.

Because of the dynamic programming principle, any $C^1$ value function of the above stopping time problem must satisfy the following HJB equation:
\begin{equation}
\label{eq: Bellman}
	rv_q(\pi) = \max\Big\{ rw(\pi), -c +\lambda_q \pi \big(1+w(1)-v_q(\pi)\big)-\lambda_q\pi(1-\pi) v'_q(\pi)\Big\}.
\end{equation} 
The above HJB is presented in the variational form, that is the first expression in the \textit{rhs} is the value of stopping --- denying the project and holding on to the outside option $w$ --- and the second expression represents the Bellman equation over the continuation region $\mathcal{M}_q$, on which $v_q(\pi) > w(\pi)$. 

The first term in the Bellman equation is the flow cost of the project borne by the agent, the second term is the expected flow of the created surplus, and the last term captures the marginal reputation loss due to the lack of success. At $\pi = 1$, where there is no learning about the self-type, that final term is absent and the above Bellman equation reduces to:
\begin{equation}
\label{eq: Bellman_at_1}
	rv_q( 1 ) = \max\Big\{ rw(1), -c +\lambda_q \big(1+w(1)-v_q(1)\big)\Big\}\,.
\end{equation}

Induced by the above stopping time problem, the matching set $\mathcal{M}$ can thus be interpreted as the continuation region for the free boundary problem of~\eqref{eq: Bellman}, namely
\begin{equation}
\label{eq: matching_set}
	\mathcal{M}=\left\{ (q,\pi) \in \{a,b\} \times [0,1]: v_q(\pi) > w(\pi)\right\},
\end{equation}
and on the \textit{stopping region} $\mathcal{M}^c$ (namely the complement of $\mathcal{M}$), the matching value function equals the agent's reputation function, i.e., $v_q(\pi)=w(\pi)$.

The overarching goal of this paper is to study the optimal outcome, which is the solution to the following fixed-point problem: the tuple $\langle w, v ,\mathcal{M}\rangle$ constitutes an optimal outcome, if (\rn{1}) given $v$ and $\mathcal{M}$, the reputation value function $w$ satisfies~\eqref{eq: reservation_value} and (\rn{2}) given $w$, the matching value function $v$ combined with the matching set $\mathcal{M}$ together solve the free boundary system of~\eqref{eq: matching_value_function} and~\eqref{eq: matching_set}. I aim to find the $C^1$ optimal value functions. 

There is in fact a two-way feedback between the reputation function $w$ and the matching variables $\langle v,\mathcal{M}\rangle$. The link connecting $w$ to $\langle v,\mathcal{M}\rangle$ is upheld by the stopping time problem~\eqref{eq: matching_value_function}. The opposite link from the matching variables to $w$ is supported by the Bellman equation for the reputation function in~\eqref{eq: reservation_value}.
% \begin{figure}[ht]
% \begin{center}
%     \begin{tikzpicture}[shorten >=1pt,node distance=2cm,on grid,>=stealth']
%     \node(a) at (0,0)  [rectangle,rounded corners,draw,very thick]   {$\langle w \rangle$};
%     \node(b) at (4,0)  [rectangle,rounded corners,draw,very thick]   {$\langle v,\mathcal{M}\rangle$}; 

%     \node(1) at (0,-0.6) {};
%     \node(2) at (4,-0.6) {};
%     \path[thick,->] (a.north east) edge [in=170,out=10] (b.north west) ;
%     \path[thick,->] (b.south west) edge [in=350,out=190] (a.south east);
%     \end{tikzpicture}
%     \caption{Endogenous feedback}
%     \label{fig: equil_feedbacks}
% \end{center}
% \end{figure}

I should emphasize that in the optimal stopping time literature the exit option is usually exogenously set, and thus finding the optimal strategy only requires solving the free boundary problem. The main stretch in our setting is that the exit option itself is endogenously determined by the value function associated with the stopping time problem, and this complicates the solution method.
%%%%%%%%%%%%%%%%%%%%%%%%%%%%%%%%%%%%%%%%%%%%%%%%%%%%%%%%%%%%%%%%%%%%%%%%%%%%%%%%%%%%%

\section{Optimum as the Fixed-Point}
\label{sec: Opt_with_learning}
In the previous section, the optimal outcome was expressed as the fixed-point to the system of necessary conditions~\eqref{eq: reservation_value}, \eqref{eq: matching_value_function} and~\eqref{eq: matching_set}.

Below in Section~\ref{sec: nec_cond}, I appeal to the fact that any $C^1$ solution to the stopping time problem of~\eqref{eq: matching_value_function} satisfies the Bellman equation~\eqref{eq: Bellman}. Additionally, it satisfies two other conditions known as the \textit{majorant} and \textit{superharmonic} properties. Hence, I initiate the search for the optimal tuple in the larger space of $C^1$ functions that satisfy the aforementioned two properties, as well as the system of conditions~\eqref{eq: reservation_value}, \eqref{eq: Bellman} and~\eqref{eq: matching_set}.

Subsequently, in Section~\ref{subs: unique_existence}, I show the predicted outcome (determined by the above necessary conditions and denoted by $\langle w^*, v^*, \mathcal{M}^* \rangle$) is unique. Then, I show this unique tuple is indeed the optimal outcome, that is replacing~\eqref{eq: matching_value_function} with~\eqref{eq: Bellman} was innocuous, and the pair $\langle v^*,\mathcal{M}^*\rangle$ solves the stopping time problem in~\eqref{eq: matching_value_function} given $w^*$.
%---------------------------------------------------
\subsection{Necessary Conditions}
\label{sec: nec_cond}
In this section, I first show the monotonicity of the matching value functions in $q$. That is to prove for any solution $v$ to the stopping time problem~\eqref{eq: matching_value_function}, one has $v_b(\pi) \geq v_a(\pi)$ for all $\pi$. Second, I highlight two additional necessary conditions, called the \textit{majorant} and \textit{superharmonic} properties, that the optimal value functions must satisfy.\footnote{These two conditions are standard in the literature of optimal stopping and can be found in Chapter~2 of \cite{peskir2006optimal}.} 
\begin{proposition}[Monotonicity]
\label{prop: q_monotonicity}
    Optimal matching value functions must satisfy $v_b(\pi) \geq v_a(\pi)$ for all $\pi$.
\end{proposition}
\noindent\textit{Proof sketch.} Since $\lambda_b >\lambda_a$, for every initial $\pi$ the success arrives faster with a $b$-project. This means when deciding to match with a $b$-project, the agent can mimic the matching strategy of an $a$-type, and guarantees herself a payoff of at least $v_a(\pi)$. The proof is actually more subtle and is presented in the appendix.
\begin{corollary}
\label{cor: set_inclusion}
At the optimum $\mathcal{M}_a \subseteq \mathcal{M}_b$ and $w(\pi)=0$ if and only if $v_b(\pi)=0$.
\end{corollary}
\begin{proof}
The justification for $\mathcal{M}_a\subseteq \mathcal{M}_b$ is that $\pi \in \mathcal{M}_a$ implies $v_a(\pi)>w(\pi)$. Since $v_b(\pi)\geq v_a(\pi)$, then $v_b(\pi)>w(\pi)$, and hence $\pi\in \mathcal{M}_b$. In regard to the second claim, note that by equation~\eqref{eq: reservation_value}, at the optimum $w$ is a linear combination of $v_a$ and $v_b$ (with possibly zero weights). Since $v_b\geq v_a$ and both are non-negative, the second claim follows.   
\end{proof}
\begin{corollary}
\label{cor: pi=1}
    At the optimum, $1\in \mathcal{M}_a$ if and only if
    \begin{equation}
    \label{eq: cost_regime_determination}
        \lambda_a-c> \frac{\kappa \varphi_b(\lambda_b-c)}{r+\kappa\varphi_b+\lambda_b}\,.
    \end{equation}
\end{corollary}
The proof follows from the previous corollary and equation~\eqref{eq: Bellman_at_1}. We say that the economy is in the \textit{low cost regime} if the above inequality holds, and otherwise is in the \textit{high cost regime}. In particular, it claims that at $\pi=1$ --- where the learning channel is absent --- selecting an $a$-project is optimal if its payoff exceeds the opportunity cost, that is induced by waiting for a superior $b$-project.
\begin{remark}
\label{rem: w(1)}
It is noteworthy to mention that $w(1)$ itself is an endogenous object, that takes different forms depending on the cost regime. Its values in the high and low cost regimes are respectively expressed in equations~\eqref{eq: high_cost_w} and~\eqref{eq: low_cost_w} of the appendix --- in the proof of Corollary~\ref{cor: pi=1}. To avoid introducing extra notation, I will henceforth use $w(1)$ to refer to both of these expressions, considering that the cost regime is clear from the context.
\end{remark}

Turning to the second group of necessary conditions, the dynamics of the reputation process can be compactly represented by 
\begin{equation*}
    \d \pi_t = \left(1-\pi_{t^-}\right)\left(\d \iota_t -\lambda_q\pi_{t^-}\d t\right)\,,
\end{equation*}
where $\iota$ is the success indicator process, that is $\iota_t :=1_{\{t \geq \sigma\}}$. The infinitesimal generator associated with this stochastic process is $\mathcal{L}_q: C^1[0,1]\to C[0,1]$, where for a generic $u \in C^1[0,1]$:\footnote{Space of continuously differentiable functions on $(0,1)$ with continuous extension to the boundary $\{0,1\}$.}
\begin{equation*}
    \left[\mathcal{L}_q u\right](\pi) = \lambda_q \pi \big(1+w(1)-u(\pi)\big) -\lambda_q \pi(1-\pi)u'(\pi).
\end{equation*}
For every candidate fixed-point tuple $\langle w, v ,\mathcal{M}\rangle$ in the space of $C^1$ functions, that satisfy the system~\eqref{eq: reservation_value}, \eqref{eq: matching_value_function} and~\eqref{eq: matching_set} the following two conditions must hold for the optimal $v$ and $w$ at all $\pi \in [0,1]$ and $q \in \{a,b\}$:
\begin{enumerate}[label=(\roman*)]
	\item \textit{Majorant property}: $v_q(\pi) \geq w(\pi)$.
	\item \textit{Superharmonic property}: $[\mathcal{L}_q v_q](\pi)-rv_q(\pi)-c\leq 0$.
\end{enumerate}
The first condition simply means that in every match the agent has the option to terminate the project, thereby enjoying her reputation value $w$. The second condition means \textit{on expectation} a typical agent \textit{loses} if she decides to keep the match on the stopping region. 

Usually in the ``one dimensional'' experimentation settings, where the continuation region is one dimensional, the agent follows threshold strategy and thus the continuation region is naturally a connected subset.\footnote{One dimensional threshold rule is ubiquitous in many optimal stopping problems, e.g., \cite{hyndman2022procrastination, henderson2023cautious}.} However, in the current setting, where the matching set is ``two-dimensional'', consisting of two sections $\mathcal{M}_a$ and $\mathcal{M}_b$, one may expect a situation in which one of these subsets contains two disjoint intervals, and hence not be connected. In the next three results, using Proposition~\ref{prop: q_monotonicity} and the above two conditions, I will rule out this possibility, and show that both sections of the matching set are connected intervals.
\begin{lemma}[Lowest boundary point]
\label{lem: beta}
Let $\beta:= \inf\mathcal{M}_b$, that is the lowest boundary point of the high type section of the optimal matching set. Then,
\begin{equation}
\label{eq: beta}
    \beta = \frac{c}{\lambda_b \big(1+w(1)\big)}\,.
\end{equation}
\end{lemma}
\begin{proof}
Since the value functions are continuous, $\mathcal{M}_b$ is an open subset in $[0,1]$, and hence $\beta \notin \mathcal{M}_b$. Corollary~\ref{cor: set_inclusion} implies that $w(\pi) = 0$ for all $\pi\leq \beta$. By optimality the matching value function $v_b$ must smoothly meet the zero function at $\beta$, i.e., $v_b(\beta)=v'_b(\beta)=w'(\beta)=w(\beta)=0$. By substituting this into the Bellman equation~\eqref{eq: Bellman}, we arrive at~\eqref{eq: beta}.
\end{proof}

The next lemma shows that at the optimum, $\mathcal{M}_b$ is an \textit{increasing interval}. That is if $\pi \in \mathcal{M}_b$, then $\pi' \in \mathcal{M}_b$ for all $\pi'>\pi$. To show this claim, suppose to the contrary that $\exists \pi' >\pi$ such that $\pi' \notin \mathcal{M}_b$. Then, Corollary~\ref{cor: set_inclusion} implies tha $v_b(\pi')=0$, whereas $v_b(\pi)>0$ because $\pi \in \mathcal{M}_b$. This combination will be ruled out in the next lemma, thus proving that $\mathcal{M}_b$ is an increasing interval.

\begin{lemma}[Single crossing]
\label{lem: single_cross}
Let $v_b$ be the optimal matching value function in $C^1$. If $v_b(\pi)>0$, then $v_b(\pi')>0$ for all $\pi'>\pi$.
\end{lemma}
\begin{proof}
Assume there exists $\pi'>\pi$ such that $v_b(\pi')=0$. Since $\pi \in \mathcal{M}_b$, then clearly $\beta < \pi<\pi'$. Also, since $v_b \geq 0$, then $\pi'$ is a global minimum and thus $v'_b(\pi')=0$. Because of the superharmonic property at $\pi'$, one has
\begin{equation*}
    0 \geq \mathcal{L}_b v_b(\pi')-rv_b(\pi') = \lambda_b \pi' \big(1+w(1)\big) -c\,,
\end{equation*}
where the equality holds because $v_b(\pi')=v'_b(\pi')=0$. This in turn implies that $\pi'\leq c/\lambda_b\big(1+w(1)\big)=\beta$, which is a contradiction.
\end{proof}
As previously argued, the above lemma shows that at the optimum $\mathcal{M}_b$ is the increasing interval $(\beta,1]$, and hence is unique. This in turn leaves only one candidate for the optimal $v_b$, that satisfies the Bellman equation~\eqref{eq: Bellman} on the continuation region $(\beta,1]$. A particular solution for this differential equation is
\begin{equation*}
    -\frac{c}{r}+\frac{\lambda_b}{r+\lambda_b}\left(1+w(1)+\frac{c}{r}\right)\pi\,,
\end{equation*}
and the homogeneous solution is $\left(1-\pi\right)^{1+r/\lambda_b}\pi^{-r/\lambda_b}$. Since $v_b(\beta)=v'_b(\beta)=0$, the only candidate for the optimal $v_b$ is 
\begin{equation}
\label{eq: optimal_vb}
	\begin{gathered}
		v^*_b(\pi) = -\frac{c}{r}+\frac{\lambda_b}{r+\lambda_b}\left(1+w(1)+\frac{c}{r}\right)\pi\\
			+\left(\frac{c}{r}-\frac{\lambda_b}{r+\lambda_b}\left(1+w(1)+\frac{c}{r}\right)\beta\right)\left(\frac{1-\pi}{1-\beta}\right)^{1+r/\lambda_b}\left(\frac{\pi}{\beta}\right)^{-r/\lambda_b}\,.
	\end{gathered}
\end{equation}
Exploiting the above characterization as the only viable candidate for the optimal $v_b$ (in any $C^1$ fixed-point outcome) and the fact that $\mathcal{M}_a\subseteq \mathcal{M}_b$, I prove in the following proposition that $\mathcal{M}_a$ is also an increasing interval --- especially, it means $\mathcal{M}_a$ cannot have disjoint subsets. The proof involves several steps, so it is relegated to the appendix.
\begin{proposition}[Optimal $\mathcal{M}_a$]
\label{prop: optimal_Ma}
    In the low cost regime, optimal $\mathcal{M}_a$ is an increasing interval, i.e., $\mathcal{M}_a = (\alpha,1]$ for some $\alpha \geq \beta$. And in the high cost regime $\mathcal{M}_a=\emptyset$.
\end{proposition}
This proposition implies that in the high cost regime $v^*_a = \frac{\kappa \varphi_b}{r+\kappa \varphi_b}\, v^*_b$, and $\mathcal{M}^*_a = \emptyset$. In the low cost regime, however, $\mathcal{M}^*_a = (\alpha,1]$ and 
\begin{equation}
\label{eq: optimal_va}
    v^*_a(\pi) = -\frac{c}{r}+\frac{\lambda_a}{r+\lambda_a}\left(1+w(1)+\frac{c}{r}\right)\pi+\gamma \left(\frac{1-\pi}{1-\alpha}\right)^{1+r/\lambda_a}\left(\frac{\pi}{\alpha}\right)^{-r/\lambda_a}\,,
\end{equation}
in that the coefficient $\gamma$ and the lower boundary point $\alpha$ are determined by the following boundary conditions: 
\begin{equation}
\label{eq: boundary_cond_va}
v^*_a(\alpha) = \frac{\kappa \varphi_b}{r+\kappa \varphi_b} \, v^*_b(\alpha)\, \text{ and } v^{*'}_a(\alpha) = \frac{\kappa \varphi_b}{r+\kappa \varphi_b} \, v^{*'}_b(\alpha)\,.
\end{equation}

%-------------------------------------------
\subsection{Uniqueness and Martingale Verification}
\label{subs: unique_existence}
The characterizations in the previous section essentially offered a unique tuple as the only viable candidate satisfying~\eqref{eq: reservation_value}, \eqref{eq: Bellman}, and~\eqref{eq: matching_set} as well as the majorant and superharmonic conditions. In the first theorem below, I summarize the properties of this tuple.
\begin{theorem}[Uniqueness]
\label{thm: unique}
    The following profile expresses the unique $C^1$ value functions and the matching sets, that satisfy the fixed-point conditions~\eqref{eq: reservation_value}, \eqref{eq: Bellman}, and ~\eqref{eq: matching_set} as well as the majorant and superharmonic properties:
    \begin{enumerate}[label=(\roman*)]
        \item In each cost regime $\mathcal{M}^*_b = (\beta,1]$, where $\beta$ is determined by Lemma~\ref{lem: beta}. Additionally, $v^*_b$ follows~\eqref{eq: optimal_vb}.

        \item In the high cost regime $\mathcal{M}^*_a =\emptyset$ and $v^*_a = \frac{\kappa \varphi_b}{r+\kappa \varphi_b}\, v^*_b$. In the low cost regime $\mathcal{M}^*_a = (\alpha,1]$, where $\alpha$ is determined by~\eqref{eq: boundary_cond_va}, and $v^*_a$ follows~\eqref{eq: optimal_va}.
    \end{enumerate}
\end{theorem}
All the claims in this theorem, except a complete verification of the superharmonic property (especially outside of the continuation region), were justified in the previous section. Therefore, it only remains to establish the superharmonic property in the appendix.

The next step is to demonstrate that the unique tuple expressed in the previous theorem does indeed correspond to the best response of the agent. Formally, one needs to prove that given $w^*$, the pair $\langle v^*,\mathcal{M}^*\rangle$ describe the optimal value function and the optimal continuation region for the stopping time problem of \eqref{eq: matching_value_function}. In the next theorem, I will apply a Martingale verification procedure to show this step.
\begin{theorem}[Unique optimum]
\label{thm: optimal_policy}
     $\langle w^*, v^*, \mathcal{M}^* \rangle$ is the unique optimal tuple in the space of $C^1$ value functions, satisfying conditions~\eqref{eq: reservation_value}, \eqref{eq: matching_value_function}, and~\eqref{eq: matching_set}.
\end{theorem}
Figure~\ref{fig: value_functions} plots the optimal value functions in the low cost regime. In particular, it demonstrates the convexity of the value functions, and shows that at the the optimum $\alpha \geq \beta$. Relatedly, Figure~\ref{fig: equil_match_set} plots the optimal matching sets $\mathcal{M}^*$ in each cost regime. As explained before, at the optimum $\mathcal{M}^*_a\subseteq \mathcal{M}^*_b$, and in the high cost regime $\mathcal{M}^*_a =\emptyset$. In light of $\mathcal{M}^*_a \subseteq \mathcal{M}^*_b$, the model offers the testable prediction that the agents who exit the economy and do not engage in further activities made their last few engagements in the high-growth projects (i.e., $b$-types).

The agent follows a cutoff strategy with respect to each type of the projects, and in particular, she shows more tolerance for failure when matched to the high type projects. The threshold strategy (equivalently, that the matching sets are increasing intervals) advances the idea that agents with higher reputation have higher tolerance for failure. In other words, the distance to the endogenous separation point ($\alpha$ or $\beta$) is greater for a more reputable agent than a less reputable one. This observation is in line with the \textit{learning theory} in economics of venture capital. Specifically, \cite{gompers1999analysis} argue that VCs learn about their \textit{post-investment ability} while they are funding startups, and the more reputable ones have higher tolerance for failure, namely they spend longer time funding their portfolio companies.

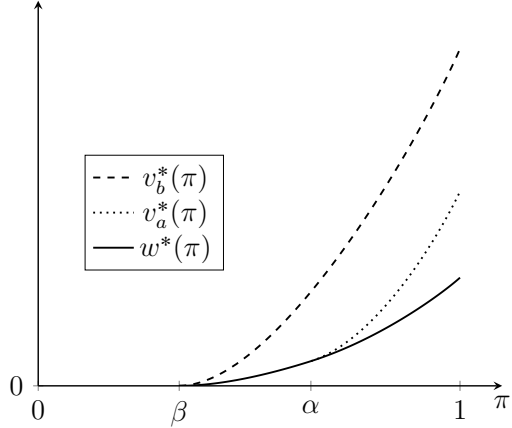
\begin{figure}[htbp]
\begin{center}
\begin{tikzpicture}[scale=0.9]
    \tikzmath{
        \r=0.9;
        \La=1.8;
        \Lb=3;
        \c=1.2;
        \Kap=0.8;
        \phia=0.35;
        \phib=0.4;
        \Ala=0.646471;
        \Alb=0.334252;
        \w=0.1967;
        \VbAla=0.170859;
        }
    \begin{axis}[
    axis lines=left,
    smooth,
    axis line style={semithick},
    xtick={0,\Ala,\Alb,1},
    xticklabels={$0$,$\alpha$,$\beta$,$1$},
    xmin=0,
    xmax=1.1,
    ytick={0},
    yticklabels={$0$},
    ytick pos=left,
    x label style={at={(axis description cs:1,0)},anchor=north},
	y label style={at={(axis description cs:0,1.13)},anchor=north,rotate=-90},
	smooth,
    xlabel=$\pi$,
    ymin=0,
    ymax=0.7,
    legend style={at={(0.1,0.6)},anchor=north west},
    ] 
    
    \addplot[thick,dashed,samples=400,domain= \Alb:1]
    {(-\c/\r)+(\Lb/(\r+\Lb))*(1+\w+(\c/\r))*x+
    ((\c/\r)-(\Lb/(\r+\Lb))*(1+\w+(\c/\r))*\Alb)*(((1-x)/(1-\Alb))^(1+(\r/\Lb)))*((\Alb/x)^(\r/\Lb))};
    \addlegendentryexpanded{$v^*_b(\pi)$}

    \addplot[thick,dotted,samples=400,domain= \Ala:1]
    {(-\c/\r)+(\La/(\r+\La))*(1+\w+(\c/\r))*x+
    (((\c/\r)-(\La/(\r+\La))*(1+\w+(\c/\r))*\Ala)+\VbAla*(\Kap*\phib/(\r+(\Kap*\phib))))*(((1-x)/(1-\Ala))^(1+(\r/\La)))*((\Ala/x)^(\r/\La))};
    \addlegendentryexpanded{$v^*_a(\pi)$}

    \addplot[thick,samples=400,domain= \Ala:1]
    {((\Kap*\phib)*((-\c/\r)+(\Lb/(\r+\Lb))*(1+\w+(\c/\r))*x+
    ((\c/\r)-(\Lb/(\r+\Lb))*(1+\w+(\c/\r))*\Alb)*(((1-x)/(1-\Alb))^(1+(\r/\Lb)))*((\Alb/x)^(\r/\Lb)))+(\Kap*\phia)*((-\c/\r)+(\La/(\r+\La))*(1+\w+(\c/\r))*x+
    (((\c/\r)-(\La/(\r+\La))*(1+\w+(\c/\r))*\Ala)+\VbAla*(\Kap*\phib/(\r+(\Kap*\phib))))*(((1-x)/(1-\Ala))^(1+(\r/\La)))*((\Ala/x)^(\r/\La))))/(\r+(\Kap*\phib)+(\Kap*\phia))};
    \addlegendentryexpanded{$w^*(\pi)$}
    
    \addplot[forget plot,thick,samples=400,domain= \Alb:\Ala]
    {(\Kap*\phib)*((-\c/\r)+(\Lb/(\r+\Lb))*(1+\w+(\c/\r))*x+
    ((\c/\r)-(\Lb/(\r+\Lb))*(1+\w+(\c/\r))*\Alb)*(((1-x)/(1-\Alb))^(1+(\r/\Lb)))*((\Alb/x)^(\r/\Lb)))/(\r+(\Kap*\phib))};    
    \end{axis}
\end{tikzpicture}
\caption{Value functions in the low cost regime}
\label{fig: value_functions}
\end{center}
\end{figure}
\begin{figure}[ht]
	\begin{subfigure}[t]{0.5\textwidth}
	    \begin{center}
		\begin{tikzpicture}[scale=1.2,thick,>=stealth',dot/.style = {draw,fill = black, circle,inner sep = 0.5pt,minimum size = 4pt}]
  			\coordinate [label=left:$0$] (O) at (0,0);
  			\coordinate [label=left:$1$] (1) at (0,3.5);
  			\draw[->] (0,0) -- (4,0) coordinate[label = {right:$\substack{\text{project}\\\text{type}}$}] (xmax);
  			\draw[->] (0,0) -- (0,4) coordinate[label = {left:$\pi$}] (ymax);
  			\draw (1,0) node[dot,label = {below:$a[ \varphi_a]$}] {}; 
  			\draw (3,0) node[dot,label = {below:$b[\varphi_b]$}] {}; 
  			\draw[dotted] (1,0) -- (1,3.5);
 			\draw[dotted] (3,0) -- (3,3.5);
  			\draw[dashed,thin] (0,3.5) -- (4,3.5);
  			\draw[line width=0.04cm] (3,1.4) -- (3,3.5);
  			\draw[)-,line width=0.04cm] (3,1.5);
  			\draw[[-,line width=0.04cm] (3,3.5);
  			\draw[dashed,thin] (0,1.4) -- (4,1.4);
  			\coordinate [label=left:$\beta_{\textsf{HC}}$] (alpha) at (0,1.4);
		\end{tikzpicture}
		\caption{\small{high cost regime}}
		\end{center}
	\end{subfigure} 
	\begin{subfigure}[t]{0.5\textwidth}
	    \begin{center}
		\begin{tikzpicture}[scale=1.2,thick,>=stealth',dot/.style = {draw,fill = black, circle,inner sep = 0pt,minimum size = 4pt}]
  			\coordinate [label=left:$0$] (O) at (0,0);
  			\coordinate [label=left:$1$] (1) at (0,3.5);
  			\draw[->] (0,0) -- (4,0) coordinate[label = {right:$\substack{\text{project}\\ \text{type}}$}] (xmax);
 			\draw[->] (0,0) -- (0,4) coordinate[label = {left:$\pi$}] (ymax);
  			\draw (1,0) node[dot,label = {below:$a[\varphi_a]$}] {}; 
  			\draw (3,0) node[dot,label = {below:$b[\varphi_b]$}] {}; 
  			\draw[dotted] (1,0) -- (1,3.5);
  			\draw[dotted] (3,0) -- (3,3.5);
  			\draw[dashed,thin] (0,3.5) -- (4,3.5);
  
  			\draw[line width=0.04cm] (3,0.9) -- (3,3.5);
  			\draw[)-,line width=0.04cm] (3,1);
  			\draw[[-,line width=0.04cm] (3,3.5);

  			\draw[line width=0.04cm] (1,2) -- (1,3.5);
  			\draw[)-,line width=0.04cm] (1,2.1);
  			\draw[[-,line width=0.04cm] (1,3.5);

  			\draw[dashed,thin] (0,2) -- (4,2); \coordinate [label=left:$\alpha_{\textsf{LC}}$] (alpha) at (0,2);
  			\draw[dashed,thin] (0,0.9) -- (4,0.9);
  			\coordinate [label=left:$\beta_{\textsf{LC}}$] (alpha) at (0,0.9);

		\end{tikzpicture}
		\caption{\small{low cost regime}}
		\end{center}
	\end{subfigure}
	\caption{Optimal matching sets}
	\label{fig: equil_match_set}
\end{figure}
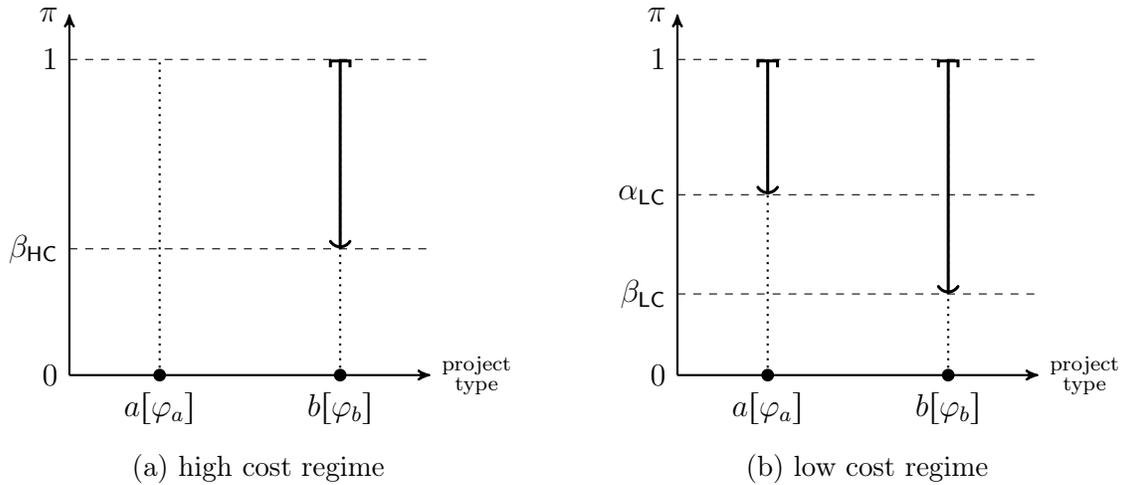
Specifically, it was shown in Lemma~\ref{lem: beta} that the endogenous termination point $\beta$ is inversely related to $w(1)$, where $w(1)$ is the value of holding the maximum reputation, namely at $\pi=1$, in each cost regime. In the high cost regime $w(1)$ only depends on the $b$-parameters, because $\mathcal{M}^*_a=\emptyset$, whereas in the low cost regime it takes the $a$-related parameters into account as well (see Remark~\ref{rem: w(1)}). Additionally, equation~\eqref{eq: beta} shows that it is indeed through the reputation channel (i.e., $w(1)$) that learning incentives manifest themselves in the agent's selection policy. Specifically, any exogenous parameter of the economy can influence $w(1)$, thereby impacting the size of the continuation set.

In the next section, I compare the current learning model with its no-learning version, and perform comparative statics with respect to the underlying primitives, in particular the meeting rate (or inversely, the search frictions).
%%%%%%%%%%%%%%%%%%%%%%%%%%%%%%%%%%%%%%%%%%%%%%%%%%%%%%%%%
\section{Qualitative Comparisons}
\label{sec: Qual}
In Section~\ref{subs: no_learning_version} below, I illustrate how the shape of the matching sets and value functions change in an economy, where the agent has complete information about her type, but is otherwise the same as before. This exercise reveals the distinct impact of learning and incomplete information on the agent's optimal selection strategy.

Subsequently, in Section~\ref{subs: comp_statics}, I perform the comparative statics of the optimal matching sets (in the original learning model) with respect to the primitives of the economy.
%--------------------------------------------

\subsection{No-Learning Version}
\label{subs: no_learning_version}
In contrast to our original model, where the agent's underlying type was a hidden binary variable $\theta\in \{L,H\}$, and $\pi$ reflected the posterior belief, here I assume the actual underlying type is $\pi \in [0,1]$, and it remains \textit{constant} over time. Specifically, when a type-$\pi$ agent selects a type-$q$ project, the success arrives with the rate of $\lambda_q \pi$. The underlying reason behind studying this benchmark case is to understand how the learning process impacts the optimal matching sets.

The major changes happen in the Bellman equation for the matching value function. First, the Bayesian learning component that includes the $\pi$-derivative of $v_q(\pi)$ is no longer present. Second, the exit option at the time of success is $1+w(\pi)$ instead of $1+w(1)$. This is due to the fact that the agent's type is persistent and she leaves the match with the same reputation she entered with. Formally, the no-learning Bellman equation for the matching value function is:
\begin{equation}
\label{eq: no_learning_Bellman}
r v_q(\pi) = \max\left\{rw(\pi),-c+\lambda_q \pi \big(1+w(\pi)-v_q(\pi)\big)\right\}\,.
\end{equation}
The expressions behind the reputation function $w$ and the set $\mathcal{M}$ remain consistent with those in~\eqref{eq: reservation_value} and~\eqref{eq: matching_set}, respectively.
\begin{proposition}[Unique optimum, absent learning]
\label{prop: no_learning_optimum}
    There exists a unique optimal tuple $\langle \hat w, \hat v , \widehat{M}\rangle$, in the space of continuous value functions, that satisfy the optimality conditions~\eqref{eq: reservation_value}, \eqref{eq: no_learning_Bellman}, and~\eqref{eq: matching_set}. Furthermore,
    \begin{enumerate}[label=(\roman*)]
        \item in both cost regimes the matching sets are increasing intervals and $\widehat{\mathcal{M}}_a \subseteq \widehat{\mathcal{M}}_b$.
        \item In the high cost regime $\widehat{\mathcal{M}}_a = \emptyset$.
    \end{enumerate}
\end{proposition}
The lower boundary of optimal matching sets (in the proof of Proposition~\ref{prop: no_learning_optimum}) are denoted by $\hat \alpha = \inf \widehat{\mathcal{M}}_a$, and $\hat \beta = \inf \widehat{\mathcal{M}}_b$. It is shown in the appendix that $\hat \beta = c/\lambda_b$ and $\hat \alpha$ follows equation~\eqref{eq: no_learning_alpha}.

\paragraph{Comparing $\hat v$ with $v^*$:} The optimal matching value functions in the current no-learning environment (and in the low cost regime) are plotted in Figure~\ref{fig: NL_value_functions}. There are two important differences with Figure~\ref{fig: value_functions}: local concavity and kinks on the boundary of the matching sets. In contrast, the value functions in the learning environment were convex and smooth. Both of these properties were due to the Bayesian learning, that is absent here. Specifically, the marginal value of acquiring information about the self-type becomes larger as the posterior belief increases, and this aspect is only present when the learning opportunity is available. Analytically, the final term in~\eqref{eq: optimal_vb}, which is induced by the Bayesian learning component in the Bellman equation, delivers the convexity of the value functions.

\begin{figure}[htp]
\begin{center}
\begin{tikzpicture}[scale=0.9]
    \tikzmath{
        \r=0.9;
        \La=1.5;
        \Lb=3.8;
        \c=0.6;
        \Kap=2.1;
        \phia=3;
        \phib=0.4;
        \Ala=0.648248;
        \Alb=\c/\Lb;
        }
    \begin{axis}[
    axis lines=left,
    smooth,
    axis line style={semithick},
    xtick={0,\Ala,\Alb,1},
    xticklabels={$0$,$\hat \alpha$,$\hat \beta$,$1$},
    xmin=0,
    xmax=1.1,
    ytick={0},
    yticklabels={$0$},
    ytick pos=left,
    x label style={at={(axis description cs:1,0)},anchor=north},
	y label style={at={(axis description cs:0,1.13)},anchor=north,rotate=-90},
	smooth,
    xlabel=$\pi$,
    ymin=0,
    ymax=1.2,
    legend style={at={(0.98,0.05)},anchor=south east},
    ] 
    
    \addplot[name path=V_b,dashed,thin,samples=400,domain= \Ala:1]
    {(\Lb*x-\c)/(\r+\Lb*x)+(\Lb*x*(\Kap*\phib*(\Lb*x-\c)*(\r+\La*x)+\Kap*\phia*(\La*x-\c)*(\r+\Lb*x))/(\r*((\r+\La*x)*(\r+\Lb*x)+\Kap*\phib*(\r+\La*x)+\Kap*\phia*(\r+\Lb*x))))/(\r+\Lb*x)};
    \addlegendentryexpanded{$\hat v_b(\pi)$}

    \addplot[name path=V_a,dotted,thick,samples=400,domain= \Ala:1]
    {(\La*x-\c)/(\r+\La*x)+(\La*x*(\Kap*\phib*(\Lb*x-\c)*(\r+\La*x)+\Kap*\phia*(\La*x-\c)*(\r+\Lb*x))/(\r*((\r+\La*x)*(\r+\Lb*x)+\Kap*\phib*(\r+\La*x)+\Kap*\phia*(\r+\Lb*x))))/(\r+\La*x)};
    \addlegendentryexpanded{$\hat v_a(\pi)$}
    
     \addplot[name path=W_ab,thick,samples=400,domain= \Ala:1]
    {(\Kap*\phib*(\Lb*x-\c)*(\r+\La*x)+\Kap*\phia*(\La*x-\c)*(\r+\Lb*x))/(\r*((\r+\La*x)*(\r+\Lb*x)+\Kap*\phib*(\r+\La*x)+\Kap*\phia*(\r+\Lb*x)))};
    \addlegendentryexpanded{$\hat w(\pi)$}
    
    \addplot[name path=V_b_alone,dashed,thin,samples=400,domain= \Alb:\Ala]
    {(\Lb*x-\c)/(\r+\Lb*x)+(\Lb*x*(\Kap*\phib*(\Lb*x-\c))/(\r*(\r+\Lb*x+\Kap*\phib))/(\r+\Lb*x))};
    
    \addplot[name path=W_b,thick,samples=400,domain= \Alb:\Ala]
    {(\Kap*\phib*(\Lb*x-\c))/(\r*(\r+\Lb*x+\Kap*\phib))};
    \end{axis}
\end{tikzpicture}
\caption{Value functions in the low cost regime (no-learning version)}
\label{fig: NL_value_functions}
\end{center}
\end{figure}
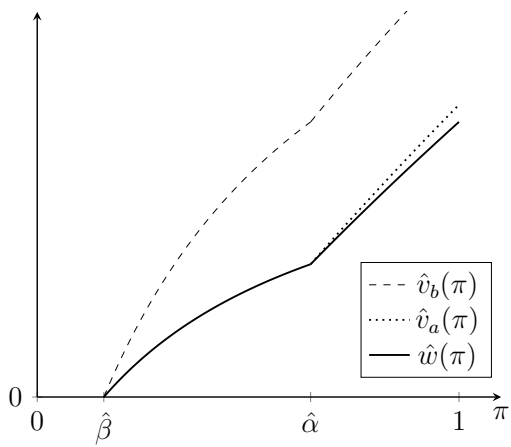

\paragraph{Comparing $\hat \beta$ with $\beta$:} It is shown in the proof of Proposition~\ref{prop: no_learning_optimum} that $\hat \beta = c/\lambda_b$. Comparing this with equation~\eqref{eq: beta} for $\beta$, one notices an important difference: learning incentives affect $\beta$ through the impact of $w(1)$ in its denominator. Specifically, increasing $\kappa,\varphi_a,\varphi_b$, or decreasing $r$ each strengthens the reputational motives and raises $w(1)$, thereby making the agent more patient (by lowering $\beta$). In the absence of learning, all of these effects are muted in $\hat \beta$. Hence, in both cost regimes, the separation point $\beta$ is smaller than its no-learning counterpart $\hat \beta = c/\lambda_b$. Therefore, the prospects of learning about the self-type and possibly reaching a higher reputation expand the matching sets and add more patience to the agent's continuation region. I refer to this force as the \textit{learning effect} in the next section.

\paragraph{Response of $\hat \alpha$ to $\kappa$:} Differentiating the expression for $\hat \alpha$ in equation~\eqref{eq: no_learning_alpha} (of the appendix) with respect to $\kappa$ implies that $\d \hat\alpha/\d \kappa>0$. Therefore decreasing the search frictions shrinks $\widehat{\mathcal{M}}_a$. Hypothetically, in a frictionless world (where $\kappa \to \infty$) the agent never selects the $a$-projects, because by equation~\eqref{eq: cost_regime_determination} its opportunity cost ($\lambda_b-c$) exceeds its payoff ($\lambda_a-c$). In reality however, search frictions create an endogenous wedge, by lowering the previous opportunity cost, and partially tilt the incentives toward the inferior $a$-projects. I refer to this force as the \textit{opportunity cost effect} in the next section. This has the same frictional spirit, by which the low-skilled individuals are selected by the employers in the labor market. Hence, it is exactly in this sense that increasing the search frictions (by lowering the meeting rate $\kappa$) expands $\widehat{\mathcal{M}}_a$. Through an example in the next section, I show that this monotone response is overturned in the original learning model of Section~\ref{sec: Opt_with_learning}  --- due to the opposing force created by the learning effect.
%---------------------------------------------------

\subsection{Comparative Statics}
\label{subs: comp_statics}
The results of this section pertains to the original model with learning. Observe that in both cost regimes $\mathcal{M}^*_a \subseteq \mathcal{M}^*_b$, and thus I take $\beta$ as a proxy for the size of the union of matching sets, i.e., $\mathcal{M}^*_a \cup \mathcal{M}^*_b$. It is important to know the comparative statics of $\beta$ (in~\eqref{eq: beta}) with respect to the primitives of the economy.

Performing simple differentiation of~\eqref{eq: beta}, one can verify that $\frac{\d \beta}{\d c} >0$, $\frac{\d \beta}{\d r}>0$, and $\frac{\d \beta}{\d \kappa} <0$. Namely, lower levels of flow cost, time discount rate, and search frictions (equivalently higher meeting rate) are all associated with larger $\mathcal{M}^*_b$. Specifically, raising the meeting rate $\kappa$ increases the value of holding the maximum reputation $w(1)$ --- because the agent meets the projects more frequently --- and this expands the optimal $\mathcal{M}^*_b$. Hence, with regard to the influence of $\kappa$ on $\beta$, only the learning effect comes into play. The opportunity cost effect has no impact on $\mathcal{M}^*_b$, as there is no better alternative than the $b$-projects. Lastly, increasing $\varphi_a$ or $\varphi_b$ raises $w(1)$ and thus decreases $\beta$ (that is similar to the effect of meeting rate $\kappa$ on $\beta$).

In the following example, I study how the optimal $\alpha$ (i.e., the lower boundary of $\mathcal{M}^*_a$) reacts \textit{non-monotonically} to the search frictions. This stands in contrast to the monotonic response of $\hat \alpha$ to $\kappa$ in the no-learning version discussed earlier.
\begin{example}[Non-monotone response of $\alpha$ to $\kappa$.]
In this example, I show --- in the low cost regime where $\mathcal{M}^*_a = (\alpha,1] \neq \emptyset$ --- there exists a range of parameters, in which the optimal $\alpha$ reacts non-monotonically to $\kappa$. This is in contrast with the response of its no-learning counterpart $\hat \alpha$ to $\kappa$, that was shown to be unambiguously increasing due to the opportunity cost effect.

First, I explain how one can mathematically pin down the fixed-point $\alpha$, and then I argue (based on the properties of the fixed-point mapping) why the response is not monotone. Observe that in the low cost regime, $\alpha$ is the point at which $v_a^*$ smoothly meets the reservation value $w^*$. By the specification in Theorem~\ref{thm: unique}, $\alpha \in \mathcal{M}^{* c}_a \cap \mathcal{M}^*_b$, and one has
\begin{equation*}
    v^*_a(\pi)=w^*(\pi) = \frac{\kappa \varphi_b}{r+\kappa \varphi_b}\, v^*_b(\pi)\,, \quad \forall \pi \leq \alpha\,.
\end{equation*}
Therefore, the boundary conditions in~\eqref{eq: boundary_cond_va} apply. Using these conditions and the Bellman equations for $v^*_a$ and $v^*_b$, one arrives at the following relation, whose fixed-point determines the optimal $\alpha$:
\begin{equation}
\label{eq: alpha_a_fixed_point}
    \alpha = \frac{r \lambda_b c+\kappa \varphi_b \big(\lambda_b-\lambda_a\big)\big(c+r v^*_b(\alpha)\big)}{r\lambda_b \lambda_a \big(1+w(1)\big)}\,.
\end{equation}
Afterward, I substituted the closed-form expression for $v^*_b$ from~\eqref{eq: optimal_vb} into the above relation. By varying $\kappa$, I found the fixed-point $\alpha$ (as a function of $\kappa$) in a numerical example whose output is plotted in Figure~\ref{fig: alpha_a_kappa}. As it appears the response is $\mathsf{U}$-shaped: for small values of $\kappa$, the optimal $\alpha$ is decreasing, while for larger $\kappa$, it becomes increasing. One should contrast this outcome with the no-learning counterpart, in which $\frac{\d \hat \alpha}{\d \kappa} >0$, and with the lower boundary of $\mathcal{M}^*_b$, where $\frac{\d \beta}{ \d \kappa} <0$.

The algebraic reason behind the $\mathsf{U}$-shaped response of $\alpha$ to $\kappa$ is that both the numerator and the denominator of the fixed-point map~\eqref{eq: alpha_a_fixed_point} are increasing in $\kappa$, therefore, the overall response is ambiguous. Intuitively however, the ratio in~\eqref{eq: alpha_a_fixed_point} highlights two opposing forces, underlying the non-monotone behavior: learning effect and the opportunity cost effect. First and similar to the case for $\mathcal{M}^*_b$, raising $\kappa$ increases the value of holding the maximum reputation $w(1)$, encouraging the agent to stay longer with the project. This learning effect manifests itself in the denominator of~\eqref{eq: alpha_a_fixed_point}, and sets an expanding force on $\mathcal{M}^*_a$. Second and similar to the case for $\widehat{\mathcal{M}}_a$, higher $\kappa$ raises the opportunity cost of choosing an $a$-project, thus shrinking the optimal $\mathcal{M}^*_a$. This effect is playing out in the numerator of~\eqref{eq: alpha_a_fixed_point}. As it appears in the example plotted in Figure~\ref{fig: alpha_a_kappa}, the learning effect dominates for small levels of the meeting rate, while as $\kappa$ increases, it is the opportunity cost effect that prevails and causes $\mathcal{M}^*_a$ to shrink.

\begin{figure}[ht]
\begin{center}
\begin{tikzpicture}[scale=0.85]
    \begin{axis}[
    axis lines=left,
    smooth,
    axis line style={thick},
    xtick={0.001},
    xticklabels={},
    ytick={0.5},
    yticklabels={},
    xmax = 1.45,
    ymin = 0.578,
    ymax = 0.59,
    x label style={at={(axis description cs:1,0)},anchor=north},
	xlabel={$\kappa$},
	y label style={at={(axis description cs:0,1)},anchor=east,rotate=-90},
	ylabel={$\alpha$},
    ]
    \addplot[thick ] table[col sep=comma] {Alpha_a_Kappa.csv};
    \end{axis}
\end{tikzpicture}
\caption{Response of the optimal $\mathcal{M}^*_a$ to $\kappa$}
\label{fig: alpha_a_kappa}
\end{center}
\end{figure}
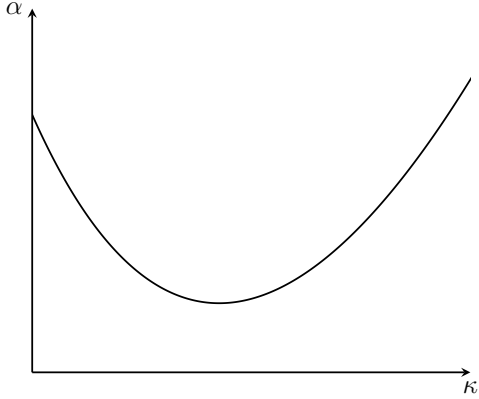
\end{example}

\begin{remark}
Even though the above observation on the non-monotone response was stated as an example, but by continuity it holds for an open region. A formal argument establishing this non-monotonicity is out of reach. First, because the fixed-point mapping in~\eqref{eq: alpha_a_fixed_point} is not monotone in the parameter $\kappa$, the monotone comparative statics apparatus cannot be applied to show the behavior in Figure~\ref{fig: alpha_a_kappa}. Second, implicitly differentiating both sides of~\eqref{eq: alpha_a_fixed_point} with respect to $\kappa$, and showing that $\frac{\d \alpha}{\d \kappa}$ is negative for small $\kappa$ and positive for large $\kappa$ is also intractable.
\end{remark}

The policy lesson behind this final comparative statics of $\mathcal{M}^*_a$ with respect to $\kappa$ is that increasing the meeting rate between the two sides of the economy in the hope of achieving higher surplus is not always socially optimal. Specifically, consider an economy where there are spillovers from successful low type projects (here the $a$-types) to the creation of high type projects (here the $b$-types) --- as is common in the innovation literature, where small low growth accomplishments create high growth opportunities. There are empirical evidences \citep{lerner2005study} suggesting that small innovative firms are particularly weak in protecting their intellectual property and thus their investors do not internalize the spillover gains in their decisions. 

In such circumstances, policies aimed at reducing the search frictions are initially helpful as they increase the incentives to invest in small projects by amplifying the learning effect. However, continued reduction eventually backfires and shrinks the investment region for the low growth projects (due to the domination of the opportunity cost effect).
%----------------------------------------------

\section{Reputational Externality}
\label{sec: rep_exter}
Building on the previous results, in this section, I study an economy populated by a continuum of agents (instead of just one) who are all making matching decisions. In frictional economies where there is not a price for reputation, one would expect agents with higher reputations to have more opportunities for contacts, which in turn suggests fewer contacts for less reputable agents. For example, in the context of two-sided market of venture capitalists and startups, there are empirical evidences about the \textit{individual benefits} associated with higher reputation among VCs. The findings include the theory of grandstanding, and lower pay-for-performance for smaller and younger VC firms toward the goal of establishing a reputation and enjoying a \textit{hig}\textit{her deal flow} \citep{gompers1996grandstanding, gompers1999analysis}.  Relatedly, by dissecting investment-level data \cite{nanda2020persistent} find that initial success confers preferential access to deal flow and perpetuates the early superior performances made by successful VCs.

Expanding upon the previous baseline results, in this section I examine how this connection between higher reputation and higher meeting rates manifests itself in the equilibrium. The main finding is that as a result of this externality, the equilibrium is not efficient, and marginally decreasing the termination threshold (thereby enlarging the equilibrium matching set) increases the social surplus of the economy. 

The nature of this reputational externality is best described if we focus only on a single group of projects. This choice is important because the potential entanglement of project selection motives (between multiple types) could complicate efforts to find equilibrium outcomes in an environment where contact rates depend on reputation. Therefore, to proceed with the analysis and clearly identify the impact of reputational externality on the equilibrium outcome, I assume that only $b$-projects are available and subsequently drop the $b$-index from the associated variables

In the following, I first find the stationary cross-sectional distribution of reputation. Then, in a mean-field setting, I postulate a functional form on how the meeting rate of each agent could depend on her reputation as well as the distribution of reputation across all agents (Section~\ref{subs: stationary_dist}). Next, in Section~\ref{subs: equil_eff}, I prove the existence of a symmetric stationary equilibrium and show that it is not surplus efficient.
%-----------------------------------------------------------

\subsection{Stationary Distribution}
\label{subs: stationary_dist}
Toward obtaining a \textit{non-degenerate} cross-sectional distribution, I assume agents are short-lived. Specifically, they leave the economy exogenously at the rate of $\delta$, and are born with the same rate, bearing the initial reputation of $p$. 

In preceding sections, the meeting rate did not depend on the agent's reputation. In this part, holding the total rate of contacts constant, I assume this flow is not uniformly distributed among agents, rather it contacts more (respectively, less) reputable agents with higher (respectively, lower) probability, according to a \textit{reputation weight} function $\psi(\cdot)$.

Let $\bm{\pi_\infty}$ be the steady-state random variable that represents the stationary distribution of reputation across all agents. Consequently, the rate at which an agent possessing a reputation of $\pi$ meets projects is 
\begin{equation*}
    \kappa \varphi \, \frac{\psi(\pi)}{\mu}\,, \text{ where } \mu:= \BE\left[\psi(\bm{\pi_\infty})\right]\,.
\end{equation*}
Here again $\kappa$ represents the extent of search frictions in the economy and $\varphi$ stands for the mass of available projects, that is exogenously replenished and held constant. I further assume $\psi$ is increasing, concave and differentiable, particularly it belongs to the following class:
\begin{equation}
\label{eq: admissible_weight_func}
    \Psi:=\Big\{\psi:[0,1]\to [0,1] \Big| \psi(0)=0,\psi(1)=1, \psi'\geq 0 \text{ and } \psi \text{ is concave}\Big\}.
\end{equation}

I conjecture (and prove in the next section) that there exists a symmetric stationary equilibrium wherein all agents terminate their matches at a common $\alpha$. In light of this conjecture, I denote the cross-sectional density function of the matched agents by $m(\pi)$ supported on $[\alpha,p]$. Let $m(1)$ and $n(1)$ be the discrete measures of the matched and unmatched agents with maximum reputation (i.e., at $\pi =1$). And finally $n(\alpha)$ and $n(p)$ are the discrete measures of unmatched group at $\alpha$ and $p$. Figure~\ref{fig: steady_state} plots all pieces of the cross-sectional steady-state distribution of agents' reputations.
\begin{figure}[htp]
\begin{center}
\begin{tikzpicture}[scale=0.85]
    \tikzmath{
        \De=0.8;
        \p=0.65;
        \r=0.9;
        \Kap=5;
        \L=2.5;
        \Fi=1;
        \Al=0.15;
        \c=1;
        \M=((((\Kap*\Fi)^2)*\L)/((\De+(\Kap*\Fi))*(\De+\L+(\Kap*\Fi))*(\De+\L)));
        \N=((\De+\L)*\M)/(\Kap*\Fi);
        \NP = \De/(\De+(\Kap*\Fi));
        \Na = ((\Kap*\Fi)/(\De+(\Kap*\Fi)))*((\Al/\p)^(\De/\L))*(((1-\Al)/(1-\p))^(-((\De/\L)+1)));
        }
    \begin{axis}[
    axis lines=left,
    smooth,
    axis line style={semithick},
    xtick={0,\Al,\p,1},
    xticklabels={$0$,$\alpha$,$p$,$1$},
    xmin=0,
    xmax=1.2,
    ytick={0},
    yticklabels={$0$},
    ytick pos=left,
    x label style={at={(axis description cs:1.1,0)},anchor=north},
	y label style={at={(axis description cs:0,1.3)},anchor=north,rotate=-90},
	smooth,
	transpose legend,
    xlabel=$\pi$,
    ymin=0,
    ymax=0.45,
    ylabel = dist. of $\bm{\pi_\infty}$,
    ] 
    
    \addplot[
    name path=MP,thick,samples=400,domain= \Al:\p]
    {(\De*\Kap*\Fi)/(\L*(\De+(\Kap*\Fi)))*((x/\p)^((\De/\L)-1))*(((1-x)/(1-\p))^(-((\De/\L)+2)))}
    [yshift=7pt]
	node [pos=0.4,above] {\small{$m(\pi)$}};

    \draw[->,thick] (axis cs:1.01,0) -- (axis cs:1.01,\M)
    [xshift=16pt]
    node [pos=0.95,above] {\small{$m(1)$}};
    
    \addplot [name path = Minline, domain=\Al:\p] {0};
    \addplot [black!10] fill between [
	of=MP and Minline,
	];
	\draw[->,thick,] (axis cs:\p,0) -- (axis cs:\p,\NP)
	[xshift=12pt]
    node [pos=0.9,above] {\small{$n(p)$}};
	
	\draw[->,thick,] (axis cs:0.99,0) -- (axis cs:0.99,\N)
	[xshift=-15pt]
    node [pos=0.85,above] {\small{$n(1)$}};
    
    \draw[->,thick,] (axis cs:\Al,0) -- (axis cs:\Al,\Na)
    [yshift=0pt]
    node [pos=1,above] {\small{$n(\alpha)$}};
    \end{axis}
\end{tikzpicture}
\caption{Steady-state cross-sectional distribution of $\bm{\pi_\infty}$}
\label{fig: steady_state}
\end{center}
\end{figure}

The inflow outflow equations at the discrete masses are:
\begin{subequations}
\label{eq: pi=1_and_p}
    \begin{align}
        \dot{m}(1) &= -\lambda m(1)+\kappa \varphi\, \frac{n(1)}{\mu}-\delta m(1),\label{eq: m1t}\\
        \dot{n}(1)&=\lambda m(1)-\kappa \varphi \frac{n(1)}{\mu}-\delta n(1)+\int_\alpha^p \lambda \pi m(\pi)\,\d \pi,\label{eq: n1t}\\
        \dot{n}(p)&=-\kappa \varphi\, \frac{\psi(p)}{\mu}\,n(p)-\delta n(p)+\delta.
        \label{eq: npt}
    \end{align}
\end{subequations}
The rationale behind the dynamics in~\eqref{eq: pi=1_and_p} is rather simple. For example, the \textit{rhs} in~\eqref{eq: n1t} consists of the influx from the successful matched agents with reputation $\pi =1$ who have just become unmatched, the outflow resulting from the recently matched individuals, the exogenous exits at the rate of $\delta$, and finally the influx stemming from successful agents across various intermediate reputation levels (from $\alpha$ to $p$).

The other discrete measure is $n(\alpha)$ (the mass of unmatched agents at the lowest reputation level, i.e., the termination point), whose value is determined by the conservation of the zeroth and first moment of the overall distribution (that is found out in the appendix).

The law of motion for the continuous density of the matched agents supported on the matching interval $(\alpha,p]$ is 
\begin{equation}
\label{eq: signle_forward}
    \dot{m}(\pi)=-\underbrace{\lambda \pi m(\pi)}_{\substack{\text{outflow of}\\ \text{successful agents}}}+\underbrace{\lambda\partial_\pi \big(\pi (1-\pi)m(\pi)\big)}_{\text{net learning inflow}}\underbrace{-\delta m(\pi)}_{\text{exogenous exits}}.
\end{equation}
The first component in the \textit{rhs} above is the outflow from $m(\pi)$ (due to the recent success events) to $n(1)$. The second term captures the net learning effect, by factoring the inflow of agents whose reputation is in $(\pi,\pi+\d \pi)$ and thus falling due to the lack of success and the outflow of the unsuccessful group with reputation in $(\pi-\d \pi,\pi)$.\footnote{The first two terms can also be understood in the context of Kolmogorov Forward equation (see Theorem~17.4.14 of \cite{cohen2015stochastic}), and associate that with the density function of the reputation process $\d \pi_t = \left(1-\pi_{t^-}\right)\left[\d \iota_t -\lambda\pi_{t^-}\d t\right]$.} Finally, the third term picks up the exogenous departures. In the steady-state $\dot{m}(\pi)=0$, hence rendering a differential equation for the density function whose solution is
\begin{equation}
\label{eq: m(pi)_sol}
    m(\pi) = m(\alpha)\left(\frac{\pi}{\alpha}\right)^{\delta/\lambda-1}\left(\frac{1-\pi}{1-\alpha}\right)^{-(\delta/\lambda+2)}, ~~\forall \pi \in [\alpha,p].
\end{equation}
The group of agents with minimum reputation at $\pi=\alpha$ are subject to two flows: the inflow from the matched individuals in $(\alpha,p]$ and the outflow due to the exogenous exits. Therefore, in the steady-state it must be that the inflow equals $\delta n(\alpha)$.
Lastly, the net inflow to the matched agents on the interval $(\alpha,p]$ is equal to the net outflow in the steady-state, that is:
\begin{equation}
\label{eq: flow_[alpha,p]}
    \underbrace{\kappa \varphi \, \frac{\psi(p)}{\mu}\,n(p)}_{\substack{\text{new matches}\\ \text{originating from }p}} = \underbrace{\lambda\int_\alpha^p \pi m(\pi)\,\d \pi}_{\substack{\text{outflow of}\\ \text{successful agents}}} +\underbrace{\delta \int_\alpha^pm(\pi)\,\d \pi}_{\text{exogenous exists}}+\underbrace{\delta n(\alpha)}_{\substack{\text{endogenously}\\ \text{separated matches}}}.
\end{equation}
Lemma~\ref{lem: steady_state_profile} in the Appendix~\ref{sec: dist_results} finds the steady-state solution to the preceding distributional equations in closed-form, thereby paving the way for the subsequent equilibrium analysis.
%-----------------------------------------------------------------------

\subsection{Equilibrium and Efficiency}
\label{subs: equil_eff}
To find the symmetric stationary equilibrium, each agent stipulates the population average for $\psi$ (denoted as $\mu$), and accordingly specifies the maximum attainable reputation function in the unmatched status (i.e., $w(1)$),   via the mapping $\mathsf{W}:[0,1] \to \BR_+$:
\begin{equation}
\label{eq: mathfrak_W}
    \mathsf{W}(\mu) := \frac{(r+\delta)^{-1}\kappa \varphi/\mu }{r+\delta+\lambda+\kappa \varphi/\mu}\,(\lambda-c)\,.
\end{equation}
Observe that the above mapping represents the value of $w(1)$ (analogous to equation~\eqref{eq: high_cost_w}), adjusted for the current setting, where the agents are short-lived and the contact rate depends on the reputation weight function $\psi(\cdot)$ and its steady-state average $\mu$.

Then, followed by the Bellman equation on the continuation region induced by $w(1)=\mathsf{W}(\mu)$, namely
\begin{equation*}
    rv(\pi) = \lambda-c + \lambda \pi \big(w(1)-v(\pi)\big)-\lambda\pi(1-\pi)v'(\pi)-\delta v(\pi),
\end{equation*}
each agent terminates the match at $\alpha=\mathsf{A}(w(1))$, where $\mathsf{A}:\BR_+\to [0,1]$ is given by $\mathsf{A}(w) := \frac{c}{\lambda(1+w)}$. This follows from the expression found for the termination point in~\eqref{eq: beta}. In the symmetric stationary equilibrium the initial stipulation about $\mu$ is self-fulfilling, that is $\mu = \mathsf{M}\big(\mu,\mathsf{A}\circ \mathsf{W}(\mu)\big)$, where $\mathsf{M}: [0,1]^2 \to \BR_+$ returns the population average of reputation weights under the steady-state measure $\bm{\pi_\infty}$, namely:
\begin{equation*}
    \mathsf{M}(\mu,\alpha) = \BE\left[\psi(\bm{\pi_\infty})\right]=m(1)+n(1)+\psi(p)n(p)+\int_\alpha^p\psi(\pi)m(\pi)\,\d \pi +\psi(\alpha)n(\alpha)\,.
\end{equation*}
\begin{definition}[Symmetric stationary equilibrium]
    The symmetric stationary equilibrium in this economy with reputational externality is the set of all fixed-points of the mapping $\mathsf{M}\big(\cdot,\mathsf{A}\circ \mathsf{W}(\cdot)\big)$ on the unit interval. A generic member is denoted by $\mu_e$. Associated with the equilibrium outcome $\mu_e$ is the equilibrium termination point $\alpha_e=\mathsf{A}\circ \mathsf{W}(\mu_e)$.
\end{definition}
\begin{theorem}
    There exists a symmetric stationary equilibrium.
\end{theorem}
\begin{proof}
In the Appendix~\ref{sec: stoch_order}, I show that an increase in $\alpha$ or $\mu$, holding the other variable constant, \textit{positively} shifts the steady-state distribution of $\bm{\pi_\infty}$ in the sense of \textit{second order stochastic dominance}. Since $\psi(\cdot)$ is increasing and concave, one can deduce that $\mathsf{M}(\mu,\alpha)$ is an increasing function in each argument. In addition, the composition map $\mathsf{A}\circ \mathsf{W}$ is increasing. Therefore the mapping $\mu \mapsto \mathsf{M}\big(\mu,\mathsf{A}\circ \mathsf{W}(\mu)\big)$ is a continuous increasing function from the unit interval to itself.\footnote{It is clearly continuous on $(0,1]$, and it is made continuous at $\mu=0$ by letting $\mathsf{W}(0):=\lim_{\mu\to 0} \mathsf{W}(\mu)$ and $\mathsf{M}(0,\alpha):=\lim_{\mu \to 0} \mathsf{M}(\mu,\alpha)$, where both limits exist in light of the expression~\eqref{eq: mathfrak_W} and Lemma~
\ref{lem: steady_state_profile} in the appendix.} Hence, a fixed-point $\mu_e$ and $\alpha_e=\mathsf{A}\circ \mathsf{W}(\mu_e)$ exist, thereby establishing the existence of a symmetric stationary equilibrium.
\end{proof}

Comparing the equilibrium outcome with the socially optimal choice, I express the steady-state flow surplus of the economy using the measures found in Lemma~\ref{lem: steady_state_profile}:
\begin{equation}
\label{eq: social_surplus_ext}
rS(\mu,\alpha) = (\lambda-c)m(1)+\int_\alpha^p(\lambda \pi-c)m(\pi)\,\d \pi.
\end{equation}
A benevolent social planner selects an $\alpha$ so that jointly with its induced $\mu$,  that is the fixed-point of $\mathsf{M}\left(\cdot,\alpha\right)$, they maximize the social surplus $S(\mu,\alpha)$.
\begin{definition}[Planner's problem]
    The planner's problem is 
    \begin{equation*}
        \max_\alpha S(\mu,\alpha) \text{ subject to } \mu=\mathsf{M}(\mu,\alpha)\,.
    \end{equation*}
\end{definition}
Note that the externality failed to be internalized in the agents' decisions is originated from the impact of their choices on $\mu$. Therefore, it is essential to incorporate $\mu = \mathsf{M}(\mu,\alpha)$ as the constraint of the planner's problem.

Next proposition explains why the identified equilibrium outcome is not constrained socially efficient, and highlights the direction along which the social surplus increases.
\begin{proposition}
\label{prop: cons_ineff_rep_ext}
    Every symmetric stationary equilibrium of the economy with reputational externality is not constrained-efficient. In particular, a local reduction in the termination point $\alpha_e$ increases the social surplus.
\end{proposition}
\begin{proof}
Every symmetric equilibrium is characterized by its associated pair $(\alpha_e,\mu_e)$, in which $\alpha_e=\mathsf{A}\circ \mathsf{W}(\mu_e)$ and $\mu_e=\mathsf{M}(\mu_e,\alpha_e)$. It is further a \textit{stable} equilibrium if $\partial_\mu \mathsf{M}(\mu_e,\alpha_e)<1$. From the expression for the social surplus in~\eqref{eq: social_surplus_ext} and Lemma~\ref{lem: steady_state_profile}, one can see that $S$ is decreasing in $\mu$, therefore, if $\mathsf{M}(\cdot,\alpha)$ has multiple fixed-points for a given $\alpha$ the one with the smallest $\mu$ is the efficient one. Furthermore, this equilibrium (with the smallest $\mu$) is stable because $\mathsf{M}(0,\alpha)>0$, and $\mathsf{M}(\cdot,\alpha)$ \textit{downcrosses} the 45-degree line at its first intersection.

Toward proving the constrained inefficiency, I employ a variational approach in the neighborhood of $\alpha_e$. Suppose the economy is in a stable pair $(\alpha_e,\mu_e)$, and the planner moves $\alpha_e$ by $\Delta \alpha$. The new smallest fixed-point $\mu_e+\Delta \mu$ satisfies
\begin{equation*}
    \mu_e+\Delta \mu = \mathsf{M}(\mu_e+\Delta \mu,\alpha_e+\Delta \alpha) \approx \mathsf{M}(\mu_e,\alpha_e)+ (\partial_\mu \mathsf{M}) \Delta \mu+(\partial_\alpha \mathsf{M})\Delta \alpha\,,
\end{equation*}
hence $\Delta \mu \approx \frac{\partial_\alpha \mathsf{M}}{1-\partial_\mu \mathsf{M}}\,\Delta \alpha$. Consequently, the change in the social surplus is:
\begin{equation}
\label{eq: inefficiency}
    r \Delta S \approx r\left(\frac{\partial_\alpha \mathsf{M}}{1-\partial_\mu \mathsf{M}} \,\partial_\mu S  +\partial_\alpha S\right)\Delta \alpha\,.
\end{equation}
Note that in every stable fixed-point of $\mathsf{M}(\cdot,\alpha_e)$, $\frac{\partial_\alpha \mathsf{M}}{1-\partial_\mu \mathsf{M}}>0$, because $\mathsf{M}$ is shown to be increasing in $\alpha$ and by the stability $\partial_\mu \mathsf{M}<1$. Additionally, as argued above $\partial_\mu S<0$. Therefore, lowering $\alpha_e$ (i.e., $\Delta \alpha<0$) leads to a strict improvement in the social surplus if $\partial_\alpha S<0$. Expression~\eqref{eq: social_surplus_ext} combined with Lemma~\ref{lem: steady_state_profile} and subsequent rearrangements, result in
\begin{equation*}
    \begin{split}
        r \partial_\alpha S(\mu_e,\alpha_e) &= (\lambda-c)\partial_\alpha m(1)-(\lambda \alpha-c)m(\alpha)\\
         &\hspace{-45pt}=-\underbrace{\frac{\kappa \varphi \psi(p)/\mu_e}{\delta +\kappa \varphi \psi(p)/\mu_e}\frac{1-p}{(1-\alpha_e)^2}\left(\frac{p}{1-p}\right)^{-\delta/\lambda}\left(\frac{\alpha_e}{1-\alpha_e}\right)^{\delta/\lambda}}_{>0} \times \\
         &\hspace{20pt} \left[\frac{\delta(\lambda \alpha_e-c)}{\lambda \alpha_e}+\frac{(\lambda-c)\kappa \varphi/\mu_e}{\delta+\lambda+\kappa \varphi/\mu_e}\right]\,.
     \end{split}
\end{equation*}
Therefore, the sign of $\partial_\alpha S(\mu_e,\alpha_e)$ is the opposite of the sign of the expression in the square bracket. Recall that in the equilibrium $\alpha_e=\mathsf{A}\circ \mathsf{W}(\mu_e)$, so
\begin{equation*}
    \begin{split}
        \frac{\delta(\lambda \alpha_e-c)}{\lambda \alpha_e}+\frac{(\lambda-c)\kappa \varphi/\mu_e}{\delta+\lambda+\kappa \varphi/\mu_e}&=-\delta \mathsf{W}(\mu_e)+\frac{(\lambda-c)\kappa \varphi/\mu_e}{\delta+\lambda+\kappa \varphi/\mu_e}\\
        &=-\delta \mathsf{W}(\mu_e)+\delta\lim_{r\to 0} \mathsf{W}(\mu_e)\geq 0,
    \end{split}
\end{equation*}
where the last inequality holds because $\mathsf{W}(\mu_e)$ is decreasing in $r$. This justifies that $\partial_\alpha S(\mu_e,\alpha_e)<0$, and hence by~\eqref{eq: inefficiency} a small reduction of equilibrium $\alpha_e$ leads to a strict improvement in the social surplus function.
\end{proof}

\begin{figure}[ht]
\begin{center}
\begin{tikzpicture}[scale=0.85]
\tikzmath{
    \P=0.4;
    \Alequil=0.280771;
    \Sequil=0.0660879;
    \Almax=0.126;
    \Smax=0.09723240985;
}
    \begin{axis}[axis lines=left,
	ymin=0,
	xmax=0.5,
	ymax=0.1,
	axis line style={thick,-stealth},
	xtick={0,\Almax,\Alequil,\P},
	xticklabels={$0$,$\alpha_*$,$\alpha_e$,$p$},
	ytick={0},
	yticklabels={$0$},
	x label style={at={(axis description cs:1,0)},anchor=north},
	xlabel={$\alpha$},
	y label style={at={(axis description cs:0,1)},anchor=east,rotate=-90},
	ylabel={$S(\alpha)$},
]
        \addplot[line width=1pt] table[col sep=comma] {Cons_surplus.csv};
    
        \draw [dashed,thin] (axis cs: \Alequil,0) --(\Alequil,\Sequil);
        
        \draw [dashed,thin] (axis cs: \Almax,0) --(\Almax,\Smax);
    \end{axis}
\end{tikzpicture}
\caption{Social surplus with reputational externality}
\label{fig: cons_surplus}
\end{center}
\end{figure}
Figure~\ref{fig: cons_surplus} is the result of a simulation that plots the social surplus as a function of $\alpha$, while implicitly satisfying $\mu=\mathsf{M}(\mu,\alpha)$ at every $\alpha\in [0,p]$. As depicted in this plot, the equilibrium termination point $\alpha_e$ is greater than the socially optimal point $\alpha_*$. Hence, the equilibrium outcome is associated with early termination of projects, and predicts a lower tolerance for failure than what is socially efficient. Given that the closed-form expressions for the steady-state distribution of $\bm{\pi_\infty}$ are intricate, establishing a non-local argument to demonstrate that $\alpha_e$ consistently exceeds $\alpha^*$ seems challenging. In the following example, I delve into an illustrative scenario that involves a limiting case.

\begin{example}[Long-lived agents, $\delta \to 0$]
According to Lemma~\ref{lem: steady_state_profile} as $\delta \to 0$, namely when agents become asymptotically long-lived, the limiting value of the flow social surplus becomes equal to
\begin{equation*}
    \lim_{\delta \to 0} rS = \lim_{\delta\to 0} (\lambda-c)m(1)=\frac{\kappa \varphi}{\lambda \mu +\kappa\varphi}\,\frac{p-\alpha}{1-\alpha}\,.
\end{equation*}
It was shown in the proof of Proposition~\ref{prop: cons_ineff_rep_ext} that $\mu$ is increasing in $\alpha$, hence the above limit is decreasing in $\alpha$. Thus, in an economy populated by long-lived agents that is subject to reputational externality, reducing the equilibrium termination point $\alpha$ unambiguously increases the social surplus.

The underlying reason for the surplus inefficiency stems from the reputational externality. This results in an undesirably high proportion of agents who have high ability yet are inactive, i.e.,  their reputation get stuck at $\alpha$. The mass of such agents is equal to $\alpha n(\alpha)$. By Lemma~\ref{lem: steady_state_profile} as $\delta \to 0$, one has
\begin{equation*}
    \lim_{\delta \to 0} \alpha n(\alpha) = \frac{\alpha(1-p)}{1-\alpha}\,,
\end{equation*}
that is increasing in $\alpha$. This means an inefficiently high proportion of agents stop the matching activity sooner than the optimal level, in spite of their high ability, which in turn reduces the social surplus.
\end{example}
%%%%%%%%%%%%%%%%%%%%%%%%%%%%%%%%%%%%%%%%%%%%%%%%%%%

\section{Concluding Remarks}
\label{sec: conclusion}
I study the optimal project selection policy of an agent with unknown ability. The agent randomly meets the projects drawn from a heterogeneous pool, that differ in their quality. In a match between the agent and a project a breakthrough arrives at the exponential rate depending on the type of the agent and the quality of the project. Since maintaining the projects are costly, the agent effectively solves a stopping time problem, in which she weighs the expected benefit of learning about her type as well as accomplishing breakthroughs against the endogenous reservation function (that is called the reputation value function in the paper).

The matching sets indicate what types of projects an agent with a certain level of reputation is willing to accept or continue the match with. 
In the space of continuously differentiable functions, I show there exists a unique optimum. Sections of the optimal matching set are increasing intervals, thus the agent follows cutoff strategies at the optimum. The thresholds depend on the type of the projects and are endogenously determined. They encode a number of messages: For example, lower levels of flow cost and time discount rate are associated with larger optimal matching sets. Additionally, it is shown that raising the meeting rate (or lowering the search frictions) has asymmetric effects across the two types of the projects: it unambiguously expands the high type section of the matching set, while on some regions it initially expands and then shrinks the low type section.

Compared to the no-learning benchmark (where there is no incomplete information about the agent's type), the optimal continuation sets are larger, therefore the agent shows more patience before stopping the projects. This is due to the convexity of the value functions in reputation, that itself is resulted from the learning incentives in the agent's dynamic problem.

Finally, the single agent setting is extended to an economy populated by a continuum of agents that exhibits reputational externality. Specifically, the meeting rate of each agent is positively impacted by her reputation and negatively by the average reputation weight across the population. Because of this externality the symmetric stationary equilibrium is not surplus efficient, and I show a local increase in the agents' tolerance (equivalently a marginal reduction in the lower endpoint of the equilibrium matching set) increases the social surplus.

%%%%%%%%%%%%%%%%%%%%%%%%%%%%%%%%%%%%%%%%%%%%%%%%%%%%%%%%%%%%%%%%
\newpage
\appendix
\counterwithin{lemma}{section}

\section{Proofs}
\label{sec: proofs}
\subsection{Heuristic Derivation of the Bellman Equations}
\label{subs: Bellman_derivation}
First, I argue how the Bellman equation for $w(\pi)$ is derived. Below, I invoke a standard dynamic programming analysis:
\begin{equation*}
	\begin{gathered}
		w(\pi) = \kappa \sum_{q \in \mathcal{M}(\pi)} \big(w(\pi)+ (v_q(\pi)-w(\pi)) \big)\varphi_q \,\d t + \kappa \sum_{q \in \{a,b\}\setminus \mathcal{M}(\pi)} w(\pi) \varphi_q \,\d t \\
	    +\big(1-\kappa (\varphi_a+\varphi_b)\,\d t\big)(1- r\,\d t)w(\pi) + o(\d t)\,.
	\end{gathered}
\end{equation*} 
The first term in the \textit{rhs} is the expected value of payoffs generated from all \textit{acceptable} matches, noting that the next project with type $q$ arrives at the rate of $\kappa \varphi_q$. The second term is the expected payoff over all \textit{denied} matches, and the third term simply refers to the discounted payoff conditioned on receiving no proposal over the period $\d t$. Rearranging the above expression and letting $\d t\to 0$ amount to the Bellman equation for $w$ in~\eqref{eq: reservation_value}.

Next, I offer a heuristic derivation of the HJB equation in~\eqref{eq: Bellman}. Observe that if the agent stops immediately, then $v_q(\pi) =w(\pi)$. On the continuation region however, assume that she keeps the match over an infinitesimal period $\d t$. Over this period she incurs the total cost of $c\, \d t$ and receives the discounted payoff that consists of two factors: (\rn{1}) with the approximate probability of $\lambda_q \pi\, \d t$ a breakthrough happens and she receives the unit prize plus the value of being an unmatched agent with the maximum reputation (i.e., $w(1)$); (\rn{2}) with the remaining probability of $1- \lambda_q \pi \,\d t$ the match does not accomplish a breakthrough, thus the agent's posterior belief drops down to $\pi - \d \pi$, and she continues with the updated valuation of $v_q(\pi- \d \pi)$. Therefore, one arrives at the following representation for $v_q(\pi)$ on the continuation region:
\begin{equation*}
    \begin{gathered}
        v_q(\pi) = -c \, \d t + (1-r\, \d t)\Big(\lambda_q \pi \,\d t\, \big(1+w(1)\big)+ (1- \lambda_q \pi \,\d t)\, v_q(\pi- \d \pi)\Big) + o(\d t)\,.
    \end{gathered}
\end{equation*}
By~\eqref{eq: Poisson_Bayes} during an unsuccessful period of $\d t$, we have $\d \pi = -\lambda_q \pi (1-\pi) \, \d t + o(\d t)$. In addition, $v_q(\pi- \d \pi) = v_q(\pi) - v'_q(\pi)\, \d \pi + o(\d t)$. Replacing these two into the above equation, applying some rearrangements and sending $\d t \to 0$ justify the HJB equation in~\eqref{eq: Bellman}.

%-------------------------------------------------
\subsection{Proof of Proposition~\ref{prop: q_monotonicity}}
Suppose the agent with reputation $\pi$ is approached by an $a$-project. Her optimal strategy is to match with the project so long as her reputation is above some level $\pi_0 \leq \pi$, i.e., the threshold rule. The case of $\pi_0=\pi$ simply means the agent rejects the project. Let $\pi_t^q$ represent the deterministic solution to equation~\eqref{eq: Poisson_Bayes}, when the agent is matched to a $q$-project. Define $t_q$ as the deterministic time at which this solution crosses the threshold $\pi_0$, namely:
\begin{equation*}
    t_q := \inf\{t\geq 0: \pi_t^q = \pi_0\}\,.
\end{equation*}
Because of Bayes law, we have
\begin{equation*}
    \frac{\pi_0}{1-\pi_0} = \frac{\pi}{1-\pi}\, e^{-\lambda_q t_q}\,.
\end{equation*}
Therefore, $\lambda_a t_a = \lambda_b t_b$, that in turn means $\BP(\sigma_a>t_a) = \BP(\sigma_b>t_b)$, in that $\sigma_q$ was defined as the exponential time of the breakthrough in a $q$-match. By the optimality of $\pi_0$ as a cutoff strategy for an $a$-match, one has:
\begin{equation*}
    \begin{gathered}
		v_a(\pi) = \BE\left[e^{-r\sigma_a}-c\int_0^{\sigma_a} e^{-rs}\d s+e^{-r \sigma_a} w(\pi_{\sigma_a}); \sigma_a \leq t_a\right]\\
		+\BE\left[-c\int_0^{t_a} e^{-r s}\d s+e^{-r t_a}w(\pi^a_{t_a});\sigma_a>t_a\right]\,.
     \end{gathered}	
\end{equation*}
Since $\pi_{\sigma_a}=1$, $\pi^a_{t_a}=\pi_0$ and $t_a$ is a deterministic time, then
\begin{equation}
\label{eq: v_a_equals}
    \begin{gathered}
        v_a(\pi) = \BE\left[e^{-r\sigma_a}-c\int_0^{\sigma_a} e^{-rs}\d s+e^{-r \sigma_a} w(1); \sigma_a \leq t_a\right]\\
        +\left(-c\int_0^{t_a} e^{-r s}\d s+e^{-r t_a}w(\pi_0)\right)\BP\left(\sigma_a>t_a\right)\,.
    \end{gathered}
\end{equation}
Recall that $v_b(\pi)$ is the \textit{optimal} matching value function when the agent is approached by a $b$-project, therefore, choosing (the deterministic value) $t_b$ as a stopping time when backing a $b$-project leads to a weakly smaller payoff. That is by equation~\eqref{eq: matching_value_function} it holds that
\begin{equation*}
    \begin{split}
		v_b(\pi) &\geq \BE\left[e^{-r\sigma_b}-c\int_0^{\sigma_b} e^{-rs}\d s+e^{-r \sigma_b} w(\pi_{\sigma_b}); \sigma_b \leq t_b\right]\\
		&\quad+\BE\left[-c\int_0^{t_b} e^{-r s}\d s+e^{-r t_b}w(\pi^b_{t_b});\sigma_b>t_b\right]\,.
	\end{split}	
\end{equation*}
By a similar reasoning, one obtains that
\begin{equation}
\label{eq: v_b_lower_bound}
    \begin{gathered}
        v_b(\pi) \geq \BE\left[e^{-r\sigma_b}-c\int_0^{\sigma_b} e^{-rs}\d s+e^{-r \sigma_b} w(1); \sigma_b \leq t_b\right]\\
		+\left(-c\int_0^{t_b} e^{-r s}\d s+e^{-r t_b}w(\pi_0)\right)\BP\left(\sigma_b>t_b\right)\,.
    \end{gathered}
\end{equation}
Next, I compare the \textit{rhs} of equations~\eqref{eq: v_a_equals} and~\eqref{eq: v_b_lower_bound}. First, observe that
\begin{equation*}
\begin{gathered}
    \BE\left[e^{-r\sigma_b}-c\int_0^{\sigma_b} e^{-rs}\d s+e^{-r \sigma_b} w(1); \sigma_b \leq t_b\right]\\
    = \BP(\sigma_b\leq t_b) \, \BE\left[e^{-r\sigma_b}-c\int_0^{\sigma_b} e^{-rs}\d s+e^{-r \sigma_b} w(1)\, \big\vert \, \sigma_b \leq t_b\right]\,.
\end{gathered}
\end{equation*}
One can easily verify that since $\lambda_b > \lambda_a$ and $\lambda_a t_a = \lambda_b t_b$, the conditional distribution of $(\sigma_a \mid \sigma_a \leq t_a)$ first order stochastically dominates the conditional distribution of $(\sigma_b \mid \sigma_b \leq t_b)$. The expression above inside the conditional expectation is a decreasing function in $\sigma$, therefore,
\begin{equation*}
\begin{gathered}
    \BE\left[e^{-r\sigma_b}-c\int_0^{\sigma_b} e^{-rs}\d s+e^{-r \sigma_b} w(1)\, \big\vert \, \sigma_b \leq t_b\right] \geq \\
    \BE\left[e^{-r\sigma_a}-c\int_0^{\sigma_a} e^{-rs}\d s+e^{-r \sigma_a} w(1)\, \big\vert \, \sigma_a \leq t_a\right]\,.
\end{gathered}
\end{equation*}
Since $\BP(\sigma_b\leq t_b)=\BP(\sigma_a\leq t_a)$, the first terms on the \textit{rhs} of~\eqref{eq: v_a_equals} and~\eqref{eq: v_b_lower_bound} compare as:
\begin{equation*}
    \begin{gathered}
        \BE\left[e^{-r\sigma_b}-c\int_0^{\sigma_b} e^{-rs}\d s+e^{-r \sigma_b} w(1); \sigma_b \leq t_b\right] \geq \\
        \BE\left[e^{-r\sigma_a}-c\int_0^{\sigma_a} e^{-rs}\d s+e^{-r \sigma_a} w(1); \sigma_a \leq t_a\right]\,.
    \end{gathered}
\end{equation*}
Regarding the second terms, observe that $\BP(\sigma_b >  t_b)=\BP(\sigma_a > t_a)$ and $t_b < t_a$, hence
\begin{equation*}
    \begin{gathered}
        \left(-c\int_0^{t_b} e^{-r s}\d s+e^{-r t_b}w(\pi_0)\right)\BP\left(\sigma_b>t_b\right) \geq \\
        \left(-c\int_0^{t_a} e^{-r s}\d s+e^{-r t_a}w(\pi_0)\right)\BP\left(\sigma_a>t_a\right)\,.
    \end{gathered}
\end{equation*}
The previous two inequalities jointly imply that $v_b(\pi) \geq v_a(\pi)$, thus proving Proposition~\ref{prop: q_monotonicity}.\qed
%-------------------------------------------------------
\subsection{Proof of Corollary~\ref{cor: pi=1}}
At $\pi=1$, there is no learning and hence the Bellman equation in~\eqref{eq: Bellman_at_1} reduces to
\begin{equation*}
    v_q(1) = \max\left\{w(1),\frac{\lambda_q -c}{r+\lambda_q }+\frac{\lambda_q  w(1)}{r+\lambda_q }\right\}\,.
\end{equation*}
This implies $1 \in \mathcal{M}_q$ if and only if $\lambda_q -c > r w(1)$. Since $\lambda_b> c$ and $\mathcal{M}_a \subseteq \mathcal{M}_b$, then $1 \in \mathcal{M}_b$ always. Let $w_b$ be the reservation value function in an outcome where $1 \notin \mathcal{M}_a$. Then according to the other leg of the fixed-point system, i.e., equation~\eqref{eq: reservation_value}, it must be that
\begin{equation}
\label{eq: high_cost_w}
    rw_b(1) = \frac{\kappa \varphi_b\left(\lambda_b-c\right)}{r+\kappa \varphi_b+\lambda_b}\,.
\end{equation}
Hence, $1\notin \mathcal{M}_a$ implies that $\lambda_a -c \leq rw_b(1)$. 

Conversely, assume $1\in \mathcal{M}_a$ and let $w_{ab}(1)$ be the reservation value function in this outcome, where $1\in \mathcal{M}_a$. Specifically, one obtains
\begin{equation}
\label{eq: low_cost_w}
	rw_{ab}(1)=\frac{\kappa \varphi_b\left(\lambda_b-c\right)\left(r+\lambda_a\right)+\kappa \varphi_a\left(\lambda_a-c\right)\left(r+\lambda_b\right)}{\left(r+\lambda_a\right)\left(r+\lambda_b\right)+\kappa \varphi_b\left(r+\lambda_a\right)+\kappa \varphi_a\left(r+\lambda_b\right)}\,.
\end{equation}
Then, $1 \in \mathcal{M}_a$ means $\lambda_a -c > rw_{ab}(1)$. Also, because of optimality in equation~\eqref{eq: reservation_value}, one has $w_{ab}(1) > w_b(1)$, hence it must be that $\lambda_a -c > rw_b(1)$.\qed

%---------------------------------------------------
\subsection{Proof of Proposition~\ref{prop: optimal_Ma}}
To prove this proposition, I will first show that $\mathcal{M}_a$ is always an interval, implying that it is always connected. By Corollary~\ref{cor: pi=1}, $1\in \mathcal{M}_a$ in the low cost regime. So the following lemma already establishes that $\mathcal{M}_a$ must be an increasing interval in the low cost regime.
\begin{lemma}
In both cost regimes (low and high) the optimal $\mathcal{M}_a$ is an interval.
\end{lemma}
\begin{proof}
Let us define $\mathcal{D}_a v_a := \mathcal{L}_a v_a -r v_a-c$. Since $v_a$ belongs to $C^1$, then $\mathcal{D}_a v_a$ is continuous. In addition, superharmonicity implies that $\mathcal{D}_a v_a(\pi)\leq 0$ for all $\pi$, and particularly, $\mathcal{D}_a v_a(\pi)=0$ on $\mathcal{M}_a$ by the Bellman equation. Assume $\pi \in \mathcal{M}_a^c \cap \mathcal{M}_b$, then 
\begin{equation*}
    v_a(\pi) = w(\pi) = \frac{\kappa \varphi_b}{r+ \kappa \varphi_b}\, v^*_b(\pi)\,.
\end{equation*}
Therefore, the Bellman equation for $v^*_b$ implies that
\begin{equation*}
    \mathcal{D}_av_a(\pi) = \frac{-\kappa \varphi_b}{r+\kappa \varphi_b}(\lambda_b-\lambda_a)\,\frac{r v^*_b(\pi)+c}{\lambda_b}+\frac{r\lambda_a \pi \big(1+w(1)\big)-cr}{r+\kappa \varphi_b}\,.
\end{equation*}
The unique characterization for $v^*_b$ in~\eqref{eq: optimal_vb} is twice differentiable. Since $\lambda_b>c$, it is easy to verify that $v_b^{* ''}  \geq 0$, and especially $v_b^{* ''} >0$ on $\mathcal{M}_b$. Hence, for $\pi \in \mathcal{M}_a^c \cap \mathcal{M}_b$ the above expression leads to
\begin{equation*}
	\frac{\d^2}{\d\, \pi^2}\, \mathcal{D}_a v_a(\pi)=\frac{-\kappa \varphi_b (\lambda_b-\lambda_a)}{\left(r+\kappa \varphi_b\right)\lambda_b}\,v_b^{* ''}(\pi)< 0\,.
\end{equation*}
Therefore, $\mathcal{D}_a v_a$ is strictly concave on every connected subset of $\mathcal{M}_a^c \cap \mathcal{M}_b$. Now assume by contradiction that $\mathcal{M}_a$ is not connected. Thus, it shall contain two disjoint maximal open intervals, say $(\pi_1,\pi_2)$ and $(\pi_3,\pi_4)$, where $\pi_2<\pi_3$. Since $\mathcal{M}_b$ is an increasing interval containing $\mathcal{M}_a$ --- respectively, by Lemma~\ref{lem: single_cross} and Corollary~\ref{cor: set_inclusion} --- it must be that $[\pi_2,\pi_3] \subset \mathcal{M}_a^c \cap \mathcal{M}_b$. The previous analysis means that $\mathcal{D}_a v_a = 0$ on $(\pi_1,\pi_2) \cup (\pi_3,\pi_4)$, and $\mathcal{D}_a v_a$ is strictly concave in between, i.e., on $[\pi_2,\pi_3]$. Thus, continuity of $\mathcal{D}_a v_a$ implies that it is positive on $[\pi_2,\pi_3]$, violating the superharmonicity, and thus proving the lemma. 
\end{proof}
The following two lemmas are aimed at proving $\mathcal{M}_a =\emptyset$ in the high cost regime. In the first one, I show a characterization for the optimal matching set that only hinges on the optimal matching value function $v$. Borrowing that in the second lemma, I show $\mathcal{M}_a$ cannot have a lower boundary point in $\mathcal{M}_b$. Thus, in light of $\mathcal{M}_a \subseteq \mathcal{M}_b$, one can conclude that $\mathcal{M}_a = \emptyset$. Finally, in both of these lemmas $v_b$ is equal to the optimal $v^*_b$ --- determined uniquely in~\eqref{eq: optimal_vb} --- but the $*$ superscript is dropped to avoid clutter.
\begin{lemma}
\label{lem: ordering}
At the optimum, $\pi \in \mathcal{M}_a \cap \mathcal{M}_b$ if and only if
\begin{equation}
\label{eq: ordering_cond}
	\frac{\kappa \varphi_a}{r+ \kappa \varphi_a} <\frac{v_b(\pi)}{v_a(\pi)} < \frac{r+\kappa \varphi_b}{\kappa \varphi_b} \,.
\end{equation} 
In addition, $\pi \in \mathcal{M}_a^c \cap \mathcal{M}_b$ if and only if the second inequality above binds.
\end{lemma}
\begin{proof}
An equivalent representation for equation~\eqref{eq: reservation_value} is 
\begin{equation}
\label{eq: chi_reservation_value}
	w(\pi) = \frac{\kappa \big(\varphi_a v_a(\pi)\chi_a(\pi) + \varphi_b v_b(\pi) \chi_b(\pi)\big)}{r+\kappa \big(\varphi_a\chi_a(\pi)+\varphi_b\chi_b(\pi)\big)}\,.
\end{equation}
One can check that if the inequality chain~\eqref{eq: ordering_cond} holds, then with $\chi_a(\pi)=\chi_b(\pi)=1$ in the above representation, both of the conditions $v_a(\pi)>w(\pi)$ and $v_b(\pi)>w(\pi)$ are satisfied, and hence the \textit{if} part is established. For the \textit{only if} direction, assume $\pi \in \mathcal{M}_a\cap \mathcal{M}_b$, then it must be that $\chi_a(\pi)=\chi_b(\pi)=1$. Replacing this in~\eqref{eq: chi_reservation_value} and simplifying $v_b(\pi)>w(\pi)$ justify the first inequality in~\eqref{eq: ordering_cond}. Similarly, simplifying $v_a(\pi)>w(\pi)$ leads to the second inequality in~\eqref{eq: ordering_cond}. The proof of the last claim in the lemma follows the same logic.  
\end{proof}
\begin{lemma}
Suppose $\mathcal{M}_a$ and $\mathcal{M}_b$ are the optimal matching sets in the high cost regime. Then, $\mathcal{M}_a$ cannot have a lower boundary point in $\mathcal{M}_b$.
\end{lemma}
\begin{proof}
Assume by contradiction that $x:=\inf \mathcal{M}_a$ belongs to $\mathcal{M}_b$. Then, continuous differentiability implies that
\begin{equation}
\label{eq: smooth_fit}
	v_a(x) = w(x)= \frac{\kappa \varphi_b}{r+\kappa \varphi_b}\,v_b(x) ~~ \text{and} ~~ v'_a(x) =w'(x)= \frac{\kappa \varphi_b}{r+\kappa \varphi_b}\,v_b'(x)\,.
\end{equation}
Now define $\Omega_q(x) := -c +\lambda_q x\big(1+w(1)\big)$ and $\Gamma_q(x):= r+\lambda_q x$. Then, continuous differentiability and the Bellman equations on the continuation regions $\mathcal{M}_a$ and $\mathcal{M}_b$ imply that:
\begin{equation*}
	\frac{v'_b(x)}{v'_b(x)}=\frac{\lambda_a}{\lambda_b}\,\frac{\Omega_b(x)-\Gamma_b(x)v_b(x)}{\Omega_a(x)-\Gamma_a(x)v_a(x)}\,.
\end{equation*}
Simplifying the previous two equations results in
\begin{equation}
\label{eq: aboundary}
	\left(\frac{\lambda_b}{\lambda_a}-1\right) rv_b(x) = -c \left(\frac{r+\kappa \varphi_b}{\kappa \varphi_b}\frac{\lambda_b}{\lambda_a}-1\right)+\frac{r\lambda_b x \big(1+w(1)\big)}{\kappa \varphi_b}\,.
\end{equation}
By Lemma~\ref{lem: ordering}, $x$ is a maximizer of $v_b(\cdot)/v_a(\cdot)$. Also, observe that $v_a$ essentially solves the same Bellman equation as $v^*_b$, albeit with different constants. Therefore, the forms of its particular and homogeneous solutions match those of $v^*_b$, and hence it becomes twice differentiable on $\mathcal{M}_a$. Since $x$ is a maximizer of $v_b/v_a$, and this ratio strictly decreases to the right of $x$, then it must be that $\lim_{\ve \downarrow 0} \left(\frac{v_b(x+\ve)}{v_a(x+\ve)}\right)''\leq 0$. Let us denote $v''_q(x):= \lim_{\ve \downarrow 0}v''_q(x+\ve)$ for $q\in \{a,b\}$. Then, the previous second order condition together with~\eqref{eq: smooth_fit} imply that:
\begin{equation}
\label{eq: local_concav}
	\frac{v''_b(x)}{v_b(x)}\leq \frac{v''_a(x)}{v_a(x)} \Rightarrow v''_b(x) \leq \frac{r+\kappa \varphi_b}{\kappa \varphi_b}\,v''_a(x)\,.
\end{equation}
One can find an expression for the second order derivatives by differentiating the Bellman equations on the continuation region:
\begin{equation*}
	rv'_q(x) = \lambda_q\big(1+w(1)-v_q(x)\big)-\lambda_q x v'_q(x)
			-\lambda_q (1-2x)v'_q(x)-\lambda_q x(1-x)v''_q(x)\,.
\end{equation*}
Replacing $v'_q$ from the original Bellman equation into the above relation yields the following expression for $v''_q$:
\begin{equation*}
    \lambda_q x(1-x)v''_q(x)=-\frac{r\big(1+w(1)\big)}{1-x}+\frac{r(r+\lambda_q)}{\lambda_q x(1-x)}\,v_q(x)+\frac{c\big(r+\lambda_q(1-x)\big)}{\lambda_q x(1-x)}\,.
\end{equation*}
Plugging the second order derivatives from above into~\eqref{eq: local_concav} and applying some rearrangements lead to the following \textit{equivalent} inequality:
\begin{equation*}
	\begin{gathered}
		rv_b(x)\left(\frac{\lambda_b}{\lambda_a}-1\right)\left(1+\frac{r}{\lambda_a}+\frac{r}{\lambda_b}\right)	\geq \\ \Big(rx\big(1+w(1)\big)-c(1-x)\Big)\left(\frac{r+\kappa \varphi_b}{\kappa \varphi_b}\frac{\lambda_b}{\lambda_a}-1\right)
		-\frac{cr}{\lambda_b}\left(\frac{r+\kappa \varphi_b}{\kappa \varphi_b}\frac{\lambda_b^2}{\lambda_a^2}-1\right)\,.
 	\end{gathered}
\end{equation*}
By substituting~\eqref{eq: aboundary} into the above inequality, and making some regroupings, one arrives at:
\begin{equation*}
x\Big[\big(1+w(1)\big)\big(\lambda_a\left(r+\lambda_b\right)-\kappa \varphi_b\left(\lambda_b-\lambda_a\right)\big)-c\left(\lambda_b+r^{-1}\kappa \varphi_b\left(\lambda_b-\lambda_a\right)\right)\Big]\geq cr\,.
\end{equation*}
In the high cost regime $w(1)$ follows~\eqref{eq: high_cost_w}, which after substitution into the above inequality leads to an equivalent condition to~\eqref{eq: local_concav}, that is only in terms of the primitives of the model:
\begin{equation}
\label{eq: primitives}
	\begin{gathered}
		\frac{cr^2}{r+\kappa \varphi_b} \left(1+\frac{\kappa \varphi_b}{r+\lambda_b}\right)+cx\lambda_b\left(1+\frac{r}{r+\lambda_b}\frac{\kappa \varphi_b}{r+\kappa \varphi_b}\right) \\
		 \leq x \big(\lambda_a\left(r+\lambda_b\right)-\kappa \varphi_b\left(\lambda_b-\lambda_a\right)\big)\,.
	\end{gathered}
\end{equation}
Next, I will show that the \textit{lhs} above is always greater than the \textit{rhs}, thus the contradiction is resulted and there is no $x=\inf \mathcal{M}_a \in \mathcal{M}_b$. Obviously at $x=0$ the \textit{lhs} is greater than the \textit{rhs}.  At $x=1$, the \textit{rhs} is increasing in $\lambda_a$, so can be upper bounded when $\lambda_a$ assumes its maximum level in the high cost regime, i.e., $c+\frac{\kappa \varphi_b\left(\lambda_b-c\right)}{r+\kappa \varphi_b+\lambda_b}$. Therefore the \textit{rhs} of~\eqref{eq: primitives} at $x=1$ is upper bounded by
\begin{equation*}
	\lambda_a\left(r+\lambda_b\right)-\kappa \varphi_b\left(\lambda_b-\lambda_a\right)\leq c\left(r+\lambda_b \right).
\end{equation*}
However, the \textit{lhs} of~\eqref{eq: primitives} equals $c(r+\lambda_b)$ at $x=1$. This mean $\mathcal{M}_a = \{1\}$, that is in contradiction with the fact that in the high cost regime $1 \notin \mathcal{M}_a$.  So~\eqref{eq: primitives} can never be satisfied. Therefore, the contradiction is resulted, and in the high cost regime $\mathcal{M}_a$ cannot have a lower boundary point in $\mathcal{M}_b$.
\end{proof}
By Corollary~\ref{cor: set_inclusion} at the optimum $\mathcal{M}_a \subseteq \mathcal{M}_b$. Thus, the previous lemma implies that $\mathcal{M}_a = \emptyset$ in the high cost regime, and thereby concluding the proof of Proposition~\ref{prop: optimal_Ma}.
%-----------------------------------------------
\subsection{Proof of Theorem~\ref{thm: unique}}
The fact that the suggested tuple satisfies the necessary conditions~\eqref{eq: reservation_value}, \eqref{eq: Bellman}, and~\eqref{eq: matching_set} as well as the majorizing condition is established in Section~\ref{sec: nec_cond}. It thus only remains to show that this tuple also satisfies the superharmonic condition.

\paragraph{Superharmonicity of $v^*_b$.} Obviously the superharmonic condition holds with equality on $(\beta,1]$, because $v^*_b$ solves the Bellman equation on this region. However, it needs to be checked on $[0,\beta]$ as well. Observe that for $\pi \in [0,\beta]$, one has $v^*_b(\pi)=0$, thus
\begin{equation*}
	\mathcal{L}_bv_b(\pi)-rv_b(\pi)-c = \lambda_b\pi \big(1+w(1)\big)-c\leq \lambda_b\beta\big(1+w(1)\big)-c=0\,,
\end{equation*}
where the last equality holds by Lemma~\ref{lem: beta}.

\paragraph{Superharmonicity of $v^*_a$.} In the low cost regime and on the interval $(\alpha,1]$, $v^*_a$ clearly satisfies the superharmonic property, because it actually solves the Bellman equation. The proofs of the superharmoincity of $v^*_a$ in the low cost regime on $[0,\alpha]$, and in the high cost regime on $[0,1]$ follow the same logic, thus here I only present the latter. In the high cost regime and on $[0,\beta]$ one has $v^*_a = 0$, and thus
\begin{equation*}	
		\mathcal{L}_a v^*_a(\pi)-rv^*_a(\pi)-c=\lambda_a\pi \big(1+w(1)\big)-c
		\leq \lambda_b \beta \big(1+w(1)\big)-c=0\,.
\end{equation*}
However, further analysis is required to establish the superharmonicity of $v^*_a$ on the interval $(\beta,1]$. On this region, $v^*_a = \frac{\kappa \varphi_b}{r+\kappa \varphi_b}\, v^*_b$, therefore,
\begin{equation*}
    \begin{gathered}
        \mathcal{L}_a v^*_a(\pi)-rv^*_a(\pi)-c = \mathcal{L}_a\left(\frac{\kappa \varphi_b}{r+\kappa \varphi_b}\,v^*_b\right)(\pi)-\frac{r\kappa \varphi_b}{r+\kappa \varphi_b}\,v^*_b(\pi)-c\\
		=\frac{\kappa \varphi_b}{r+\kappa \varphi_b}\big(\mathcal{L}_a v^*_b(\pi)-rv^*_b(\pi)-c\big)
		+\frac{r\lambda_a\pi}{r+\kappa \varphi_b}\big(1+w(1)\big)-\frac{rc}{r+\kappa \varphi_b}\,.
    \end{gathered}
\end{equation*}
Adding and subtracting $\mathcal{L}_b v^*_b$ from the first parentheses above, and observing the Bellman equation for $v^*_b$ result in
\begin{equation*}
    \begin{gathered}
        \mathcal{L}_a v^*_a(\pi)-rv^*_a(\pi)-c = -\frac{\kappa \varphi_b}{r+\kappa \varphi_b}(\mathcal{L}_b-\mathcal{L}_a)v^*_b(\pi)+\frac{r\lambda_a\pi}{r+\kappa \varphi_b}\big(1+w(1)\big)-\frac{rc}{r+\kappa \varphi_b}\\
        =-\frac{\kappa \varphi_b}{r+\kappa \varphi_b}\left(\lambda_b-\lambda_a\right)\pi \big(1+w(1)-v_b^*(\pi)-(1-\pi)v^{*'}_b(\pi)\big)
		\\
        +\frac{r\lambda_a\pi}{r+\kappa \varphi_b}\big(1+w(1)\big)-\frac{rc}{r+\kappa \varphi_b}\,.
    \end{gathered} 
\end{equation*}
One can readily verify that $v^*_b$ is convex, so, $v^*_b(\pi)+(1-\pi)v^{*'}_b(\pi) \leq v^*_b(1)$. That in turn implies an upper bound on the above expression:
\begin{equation*}
	\begin{gathered}
		\mathcal{L}_a v^*_a(\pi)-rv^*_a(\pi)-c \\
  \leq -\frac{\kappa \varphi_b}{r+\kappa \varphi_b}\frac{r\left(\lambda_b-\lambda_a\right)\pi}{r+\lambda_b}\left(1+w(1)+\frac{c}{r}\right)+\frac{r\lambda_a\pi\big(1+w(1)\big)-cr}{r+\kappa \varphi_b}
			\\
		\leq \left(-\frac{\kappa \varphi_b}{r+\kappa \varphi_b}\frac{r\left(\lambda_b-\lambda_a\right)}{r+\lambda_b}\left(1+w(1)+\frac{c}{r}\right)+\frac{r\lambda_a\big(1+w(1)\big)-rc}{r+\kappa \varphi_b}\right)^+\,.
	\end{gathered}
\end{equation*}
In the second inequality above, I used the fact that the \textit{rhs} of the first inequality is affine in $\pi$ and negative at $\pi=0$. Let us denote the argument of $(\cdot)^+$ by $\mathfrak{Z}$. It is increasing in $\lambda_a$, hence can be bounded above when $\lambda_a$ is replaced by $c+rw(1)$ (i.e., its maximum value in the high cost regime):
\begin{equation*}
	\begin{gathered}
		\mathfrak{Z}\leq -\frac{\kappa \varphi_b}{r+\kappa \varphi_b}\frac{r\big(\lambda_b-c-rw(1)\big)}{r+\lambda_b}\left(1+w(1)+\frac{c}{r}\right)+\frac{r\big(c+rw(1)\big)\big(1+w(1)\big)-rc}{r+\kappa \varphi_b} \\
		=-\frac{\kappa \varphi_b}{r+\kappa \varphi_b}\frac{\left(\lambda_b-c\right)\left(r+\lambda_b\right)\left(r+\kappa \varphi_b+c\right)}{r\left(\kappa \varphi_b+r+\lambda_b\right)^2}\\
  +\frac{\kappa \varphi_b}{r+\kappa \varphi_b}\frac{\left(\lambda_b-c\right)\left(r+\lambda_b\right)\left(r+\kappa \varphi_b+c\right)}{r\left(\kappa \varphi_b+r+\lambda_b\right)^2}=0\,.
	\end{gathered}
\end{equation*}
In the second line above $w(1)$ is replaced from equation~\eqref{eq: high_cost_w}. This concludes the proof of superharmonicity of $v^*_a$ with respect to $\mathcal{L}_a$ on $(\beta,1]$, thereby on the entire $[0,1]$.\qed
%----------------------------------------------
\subsection{Proof of Theorem~\ref{thm: optimal_policy}}
Let $w^*$ be the reputation function that is induced by $v^*$ and $\mathcal{M}^*$ following equation~\eqref{eq: reservation_value}. To verify the optimality, I need to show that given $w^*$, the matching set $\mathcal{M}^*$ is the optimal continuation region, and $v^*$ is the optimal value function for the stopping time problem of~\eqref{eq: matching_value_function}. I apply a Martingale verification argument to establish the optimality. Define $\mathbf{v}_q(\iota,\pi):=v^*_q(\pi)1_{\{\iota=0\}}+\big(\iota+w^*(\pi)\big)1_{\{\iota=1\}}$, where $\iota$ is the success indicator process. Since $\mathbf{v}$ is a bounded function, for each $q\in \{a,b\}$, one can find a bounded (and hence uniformly integrable) Martingale process $M^q$ such that:
\begin{equation}
\label{eq: Ito_rule}
	e^{-r t}\mathbf{v}_q(\iota_t,\pi_t)=\mathbf{v}_q(\iota,\pi)+\int_0^t e^{-r s}\big(\mathbb{L}_q \mathbf{v}_q(\cdot,\cdot)-r\mathbf{v}_q(\cdot,\cdot)\big)(\iota_{s^-},\pi_{s^-})\,\d s+M^q_t\,,
\end{equation}
where $\mathbb{L}_q \mathbf{v}_q(\iota,\pi):=\mathcal{L}_q v_q(\pi)1_{\{\iota=0\}}$. By the majorizing condition, for every stopping time $\tau$, one has $\mathbf{v}_q(\iota_\tau,\pi_\tau)\geq \iota_\tau+w^*(\pi_\tau)$, therefore
 \begin{equation*}
	\begin{split}
		e^{-r \tau}\big(\iota_\tau+w^*(\pi_\tau)\big)&\leq \mathbf{v}_q(\iota,\pi)+\int_0^\tau e^{-r s}\big(\mathbb{L}_q \mathbf{v}_q(\cdot,\cdot)-r\mathbf{v}_q(\cdot,\cdot)\big)(\iota_{s^-},\pi_{s^-})\,\d s+M^q_\tau\\
		&\leq \mathbf{v}_q(\iota,\pi)+\int_0^\tau  c e^{-r s}\,\d s+M^q_\tau\,,
	\end{split}
\end{equation*}
where in the second inequality I used the superharmonic property proven in Theorem~\ref{thm: unique}. Doob's optional stopping theorem implies that $\BE M_\tau^q=0$, hence for every stopping time $\tau$, one has
\begin{equation*}
	\mathbf{v}_q(\iota,\pi) \geq \BE_{\pi,q,\iota}\left[e^{-r \tau}\big(\iota_\tau +w^*(\pi_\tau)\big)-c\int_0^\tau e^{-r s }\,\d s \right].
\end{equation*}
That in turn implies
\begin{equation}
\label{eq: upper_bound_v^*}
	v^*_q(\pi)\geq \sup_\tau \BE_{\pi,q,\iota=0}\left[e^{-r \tau}\big(\iota_\tau +w^*(\pi_\tau)\big)-c\int_0^\tau e^{-r s }\,\d s \right]\,.
\end{equation}
To show the achievability in the above inequality, define
\begin{equation*}
    \tau_q:=\inf\left\{t\geq 0: \pi_t\notin \mathcal{M}^*_q \text{ or } \iota_t=1\right\}\,,
\end{equation*}
where $\mathcal{M}^*$ is characterized in Theorem~\ref{thm: unique}. Applying equation~\eqref{eq: Ito_rule} yields
\begin{equation*}
	\begin{split}
		e^{-r \tau_q}\left(\iota_{\tau_q}+w^*(\pi_{\tau_q})\right)&=e^{-r \tau_q}\mathbf{v}_q(\iota_{\tau_q},\pi_{\tau_q})\\
		&=\mathbf{v}_q(\iota,\pi)+\int_0^{\tau_q} e^{-r s}\big(\mathbb{L}_q \mathbf{v}_q(\cdot,\cdot)-r\mathbf{v}_q(\cdot,\cdot)\big)(\iota_{s^-},\pi_{s^-})\,\d s+M^q_{\tau_q}\\
		&=\mathbf{v}_q(\iota,\pi)+\int_0^{\tau_q} c e^{-r s}\,\d s+M^q_{\tau_q}\,,
	\end{split}
\end{equation*}
where the third equality holds by the Bellman on $\mathcal{M}^*_q$. Taking expectations of both sides in the above equality leads to
\begin{equation*}
    \mathbf{v}_q(\iota,\pi)=\BE_{\pi,q,\iota}\left[e^{-r \tau_q}\left(\iota_{\tau_q}+w^*(\pi_{\tau_q})\right)-c\int_0^{\tau_q}  e^{-r s}\,\d s\right].
\end{equation*}
Therefore,
\begin{equation*}
    v^*_q(\pi)=\BE_{\pi,q,\iota=0}\left[e^{-r \tau_q}\left(\iota_{\tau_q}+w^*(\pi_{\tau_q})\right)-c\int_0^{\tau_q}  e^{-r s}\,\d s\right]\,,
\end{equation*}
which together with~\eqref{eq: upper_bound_v^*} conclude the proof.\qed
%----------------------------------------------------------

\subsection{Proof of Proposition~\ref{prop: no_learning_optimum}} 
All the arguments below are stated without using \^{} on top of the variables. After I have established the unique existence of the optimum, one can reintroduce the \^{} superscript.

Observe that the variational Bellman equation for $v_q(\pi)$ in~\eqref{eq: no_learning_Bellman} can be equivalently expressed by:
\begin{equation}
\label{eq: alternate_no_learning_Bellman}
    v_q(\pi) = \max\left\{w(\pi),\frac{\lambda_q \pi-c}{r+\lambda_q \pi}+\frac{\lambda_q \pi w(\pi)}{r+\lambda_q \pi}\right\}\,.
\end{equation}
This representation implies that $v_q(\pi)>w(\pi)$ if and only if the second maximand is larger than the first, that happens if and only if $rw(\pi) < \lambda_q \pi-c$. Hence,
\begin{equation}
\label{eq: no_learning_M_determination}
    \pi \in \mathcal{M}_q \Leftrightarrow v_q(\pi) > w(\pi) \Leftrightarrow \lambda_q \pi-c > rw(\pi)\,.
\end{equation}
This already implies that in any optimal outcome, $\mathcal{M}_a \subseteq \mathcal{M}_b$, thereby proving the last claim in part (\rn{1}). Additionally, it shows that 
\begin{equation}
\label{eq: w_zero_v_zero}
    w(\pi)=0 \text{ if and only if } v_b(\pi)=0\,.
\end{equation}
Let $\hat\beta = \inf \mathcal{M}_b$, then $v_b(\hat \beta)=0$, and hence $\hat \beta = c/\lambda_b$. Then,~\eqref{eq: no_learning_M_determination} and~\eqref{eq: w_zero_v_zero} together imply that $v_b(\pi)>0$ for all $\pi>\hat \beta$, thereby proving the unique existence of $\widehat{\mathcal{M}}_b = (\hat{\beta},1]$, where $\hat \beta = c/\lambda_b$.

To show the uniqueness of $\mathcal{M}_a$, denote the reputation function in the outcome where $\pi \in \mathcal{M}_a \cap \mathcal{M}_b$ by $w_{ab}(\pi)$, and in the outcome where $\pi \in \mathcal{M}_a^c \cap \mathcal{M}_b$ by $w_b(\pi)$. One can find closed-form expressions for both of these functions using~\eqref{eq: reservation_value} and~\eqref{eq: alternate_no_learning_Bellman}. Subsequently, through extensive but straightforward algebraic computations, one can demonstrate the existence of a threshold $\hat \alpha > 0$, which makes the following statements equivalent:
\begin{equation*}
    w_{ab}(\pi) >  w_b(\pi) \Leftrightarrow \lambda_a \pi -c > rw_{ab}(\pi) \Leftrightarrow \pi > \hat \alpha\,.
\end{equation*}
In particular, the expression for $\hat \alpha$ is
\begin{equation}
    \label{eq: no_learning_alpha}
    \begin{gathered}
      \hat \alpha = \frac{1}{2\lambda_a \lambda_b}\big(c\lambda_b -r \lambda_a+\kappa \varphi_b(\lambda_b-\lambda_a)\big) \, \times \\
      \sqrt{r^2 \lambda_a^2+2r\lambda_a \big(c\lambda_b-\kappa \varphi_b(\lambda_b-\lambda_a)\big)+\big(c\lambda_b+\kappa \varphi_b(\lambda_b-\lambda_a)\big)^2}  \,.
    \end{gathered}
\end{equation}
Further computations confirm that $\hat \alpha \leq 1$ if and only if the economy is in the low cost regime. Therefore, one obtains the unique existence of $\widehat{\mathcal{M}}_a = \emptyset$ in the high cost regime, and $\widehat{\mathcal{M}}_a = (\hat\alpha,1]$ in the low cost regime. This concludes the proof of both parts in Proposition~\ref{prop: no_learning_optimum}.\qed
%%%%%%%%%%%%%%%%%%%%%%%%%%%%%%%%%%%%%%%%%%%%%%%%%%%%%%%%%%
\subsection{Required Results for Section~\ref{subs: stationary_dist}}
\label{sec: dist_results}
\subsubsection{Steady-State Distribution}
\begin{lemma}
\label{lem: steady_state_profile}
In the steady-state of the economy with short-lived agents and reputational externality, the following holds:
\begin{subequations}
    \begin{align}
        \int_\alpha^pm(\pi)\,\d \pi&=\frac{\kappa \varphi \psi(p)/\mu}{\delta+\kappa \varphi \psi(p)/\mu}\frac{p-\alpha}{\Upsilon_2(\alpha)-\alpha \Upsilon_1(\alpha)}\left(\Upsilon_1(\alpha)-\frac{\lambda}{\delta+\lambda}\Upsilon_2(\alpha)\right)\,, \label{eq: zero_mom}\\
        \int_\alpha^p \pi m(\pi)\, \d \pi &= \frac{\kappa \varphi \psi(p)/\mu}{\delta+\kappa \varphi \psi(p)/\mu}\frac{\delta}{\delta+\lambda}\frac{(p-\alpha)\Upsilon_2(\alpha)}{\Upsilon_2(\alpha)-\alpha \Upsilon_1(\alpha)}\,,\label{eq: first_mom}\\
        m(1)&=\frac{\kappa \varphi \psi(p)/\mu}{\delta+\kappa \varphi \psi(p)/\mu}\frac{\kappa \varphi/\mu}{\delta+\lambda+\kappa \varphi/\mu}\frac{\lambda}{\delta+\lambda}\frac{(p-\alpha)\Upsilon_2(\alpha)}{\Upsilon_2(\alpha)-\alpha\Upsilon_1(\alpha)}\,,\label{eq: steady_state_m(1)}\\
        n(\alpha)&=\frac{\kappa \varphi \psi(p)/\mu}{\delta+\kappa \varphi \psi(p)/\mu}\frac{\Upsilon_2(\alpha)-p\Upsilon_1(\alpha)}{\Upsilon_2(\alpha)-\alpha \Upsilon_1(\alpha)}\,,\label{eq: steady_state_n(a)}
    \end{align}
\end{subequations}
where 
\begin{equation*}
    \begin{split}
        \Upsilon_i(\alpha) := \left(\frac{p}{\alpha}\right)^{\delta/\lambda-1}\left(\frac{1-p}{1-\alpha}\right)^{-(\delta/\lambda+2)}p^i(1-p)-\alpha^i(1-\alpha), ~\text{ for } i\in \{1,2\}\,.
    \end{split}
\end{equation*}
\end{lemma}
\newpage
\begin{proof}
In the steady-state the time derivatives in~\eqref{eq: pi=1_and_p} are zero, therefore~\eqref{eq: m1t} and~\eqref{eq: n1t} jointly result in:
\begin{subequations}
    \begin{align}
        n(1)&=\frac{\delta+\lambda}{\kappa \varphi/\mu}\,m(1)\,,\\
        (\delta+\lambda+\kappa \varphi/\mu)\delta m(1)&= \kappa \varphi/\mu \int_\alpha^p\lambda \pi m(\pi)\,\d \pi\,.\label{eq: m(1)_first_mom}
    \end{align}
\end{subequations}
Also at $\pi=p$, equation~\eqref{eq: npt} implies that $n(p) = \delta/(\delta+\kappa \varphi\psi(p)/\mu)$. Next, the expression found in~\eqref{eq: m(pi)_sol} translates to
\begin{equation}
\label{eq: pim(pi)integral}
    \int_\alpha^p \pi m(\pi)\,\d \pi = \frac{\lambda m(\alpha)}{\delta+\lambda}\Upsilon_2(\alpha)\,.
\end{equation}
The \textit{rhs} to~\eqref{eq: flow_[alpha,p]} can be simplified using the steady-state differential equation resulted from $\dot{m}(\pi)=0$:
\begin{equation}
\label{eq: ma_na_1}
    \begin{gathered}
        \kappa \varphi\frac{\psi(p)}{\mu}\,n(p)=\frac{\delta\kappa \varphi\psi(p)/ \mu }{\delta+\kappa \varphi \psi(p)/\mu} = \lambda\int_\alpha^p \pi m(\pi)\,d \pi +\delta \int_\alpha^pm(\pi)\,\d \pi +\delta n(\alpha)\\
        =\lambda \int_\alpha^p\partial_{\pi}\big(\pi(1-\pi)m(\pi)\big)\,\d \pi+\delta n(\alpha)\\
        =\lambda\big(p(1-p)m(p)-\alpha (1-\alpha)m(\alpha)\big)+\delta n(\alpha)\\
        =\lambda m(\alpha)\Upsilon_1(\alpha)+\delta n(\alpha)\,.
    \end{gathered}
\end{equation}
Because of the Bayesian learning during the matches the steady-state average reputation must be equal to $p$ (i.e., conservation of the first moment):
\begin{equation*}
    m(1)+n(1)+pn(p)+\int_\alpha^p\pi m(\pi)\, \d \pi+\alpha n(\alpha)=p\,.
\end{equation*}
Simplifying this relation using~\eqref{eq: m(1)_first_mom} and~\eqref{eq: pim(pi)integral} implies that 
\begin{equation}
\label{eq: ma_na_2}
    \lambda m(\alpha) \Upsilon_2(\alpha)+\alpha \delta n(\alpha)=\frac{\delta \kappa \varphi\psi(p)/\mu }{\delta+\kappa \varphi \psi(p)/\mu}\,p\,.
\end{equation}
It is now straightforward to solve for~$n(\alpha)$ and~$m(\alpha)$ using~\eqref{eq: ma_na_1} and~\eqref{eq: ma_na_2}, thereby obtaining~\eqref{eq: steady_state_n(a)} and 
\begin{equation}
\label{eq: m(alpha)_expression}
    m(\alpha) = \frac{\delta \kappa \varphi \psi(p)/\mu}{\lambda\left(\delta+\kappa \varphi \psi(p)/\mu\right)}\frac{p-\alpha}{\Upsilon_2(\alpha)-\alpha \Upsilon_1(\alpha)}\,.
\end{equation}
Substituting $m(\alpha)$ from above into~\eqref{eq: pim(pi)integral} yields the lemma's claim for $\int_{\alpha}^p\pi m(\pi)\,\d \pi$, i.e. equation~\eqref{eq: first_mom}. Subsequently, $m(1)$ can be found from~\eqref{eq: m(1)_first_mom}, thereby verifying~\eqref{eq: steady_state_m(1)}. Finally, from the second line in~\eqref{eq: ma_na_1}, one obtains the following expression
\begin{equation*}
    \int_\alpha^pm(\pi)\,\d \pi = \frac{\lambda m(\alpha)}{\delta}\left(\Upsilon_1(\alpha)-\frac{\lambda}{\delta+\lambda}\Upsilon_2(\alpha)\right)\,,
\end{equation*}
that leads to~\eqref{eq: zero_mom} after replacing $m(\alpha)$ in the above expression.
\end{proof}
%---------------------------------------------------------------

\subsubsection{Stochastic Ordering}
\label{sec: stoch_order}
For a clearer understanding of the stochastic ordering of the steady-state distribution $\bm{\pi_\infty}$, I will start by formulating the CDF of the density $m(\cdot)$:
\begin{equation}
\label{eq: m_CDF}
    \begin{gathered}
        \int_\alpha^\pi m(x)\,\d x =\frac{\kappa \varphi \psi(p)/\mu}{\delta + \kappa \varphi \psi(p)/\mu}\frac{p-\alpha}{\Upsilon_{2,1}(\alpha,p)-\alpha \Upsilon_{1,1}(\alpha,p)}\Big(\frac{\delta}{\delta+\lambda}\Upsilon_{2,1}(\alpha,\pi)+\Upsilon_{1,2}(\alpha,\pi)\Big),\\
        \Upsilon_{i,j}(x,y):=\left(\frac{y}{x}\right)^{(\delta/\lambda-1)} \left(\frac{1-y}{1-x}\right)^{-(\delta/\lambda+2)}y^i(1-y)^j-x^i(1-x)^j\,.
    \end{gathered}
\end{equation}
In addition, using the solution found for $m(\pi)$ and the expression~\eqref{eq: m(alpha)_expression} for $m(\alpha)$, it is easy to verify that for $\pi \in [\alpha,p]$:
\begin{equation}
\label{eq: m(pi)_ind_alpha}
    m(\pi) = \frac{\delta \kappa \varphi \psi(p)/\mu}{\lambda\left(\delta+\kappa \varphi \psi(p)/\mu\right)}\left(\frac{\pi}{p}\right)^{(\delta/\lambda-1)}\left(\frac{1-\pi}{1-p}\right)^{-(\delta/\lambda+2)}\frac{1}{p(1-p)}\,.
\end{equation}
Therefore for a fixed $\mu$ the above density is \textit{independent} of $\alpha$. 

My next goal is to show that $\mathsf{M}(\mu,\alpha)$ is increasing in each argument holding the other one constant. To achieve this, I turn to the theory of stochastic orders, and in particular, I employ the second order stochastic dominance. For two real-valued random variables $X$ and $Y$, it is said that $X \succeq_{\mathsf{SSD}} Y$ if $\BE u(X) \geq \BE u(Y)$ for every increasing and concave function $u$. An equivalent definition is that $X \succeq_\mathsf{SSD} Y$ if $\BE\left[(X-t)_-\right] \geq \BE\left[(Y-t)_-\right]$ for every $t \in \BR$, provided that the expectations exist.\footnote{For every $r\in \BR$, $(r)_-:=\min\{r,0\}$. The reader can refer to Chapter~4 of \cite{shaked2007stochastic} for the proof of the equivalence.} The following lemma provides a sufficient condition for the second order stochastic dominance, drawing from the results of \cite{karlin1963generalized}.
\newpage
\begin{lemma}[Sufficient condition for \textsf{SSD}]
\label{lem: suff_SSD}
    Suppose that the following two conditions hold: (\rn{1}) $\BE[X] \geq \BE[Y]$, and (\rn{2}) there exists $t_0\in \BR$ such that for all $t\leq t_0$, $\BP\left(X \geq t\right) \geq \BP\left(Y\geq t\right)$, and for all $t>t_0$, $\BP\left(X\geq t\right) \leq \BP\left(Y \geq t\right)$. Then $X \succeq_\mathsf{SSD} Y$.
\end{lemma}
\begin{proof}
For every $t \leq t_0$,
\begin{equation*}
    \begin{gathered}
        \BE\left[(X-t)_-\right] =- \int_0^\infty \BP\left(-(X-t)_->u\right)\,\d u= -\int_0^\infty \BP\left(X< t-u\right)\,\d u\\
        =-\int_{-\infty}^t \BP\left(X<z\right)\, \d z
        \geq -\int_{-\infty}^t \BP\left(Y<z\right) \,\d z=\BE\left[(Y-t)_-\right]\,.
    \end{gathered}
\end{equation*}
Also, an equivalent representation for $\BE\left[(X-t)_-\right]$ is
\begin{equation*}
    \begin{gathered}
        \BE\left[(X-t)_-\right]=\BE\left[(X-t);X<t\right]\\
        =\BE[X]-t-\BE\left[X-t;X\geq t\right]=\BE[X]-t-\int_t^\infty \BP\left(X \geq z\right)\,\d z\,.
    \end{gathered}
\end{equation*}
Therefore, 
\begin{equation*}
    \begin{split}
        \BE\left[(X-t)_-\right]-\BE\left[(Y-t)_-\right]=\BE[X]-\BE[Y]+\int_t^\infty \big(\BP\left(Y \geq z\right)-\BP\left(X \geq z\right) \big)\,\d z\,.
    \end{split}
\end{equation*}
The first term is positive and the integral term is also positive for all $t>t_0$, so $\BE\left[(X-t)_-\right] \geq \BE\left[(Y-t)_-\right]$ for all $t>t_0$ as well.
\end{proof}
I will use the technique offered in this lemma to prove that an increase in $\alpha$ or $\mu$ \textit{positively} shifts the steady-state distribution of $\bm{\pi_\infty}$. This distribution is completely described by the measures found in Lemma~\ref{lem: steady_state_profile}.
\begin{lemma}
    Let $\alpha_1\leq \alpha_2<p$ and $\mu_1\leq \mu_2$, then
    \begin{enumerate}[label=(\roman*)]
        \item \label{item: part_mu_mon} holding $\alpha$ constant, $\bm{\pi_\infty}(\mu_2) \succeq_{\mathsf{SSD}} \bm{\pi_\infty}(\mu_1)$.
        \item \label{item: part_alpha_mon} Holding $\mu$ constant, $\bm{\pi_\infty}(\alpha_2) \succeq_{\mathsf{SSD}} \bm{\pi_\infty}(\alpha_1)$.
    \end{enumerate}
\end{lemma}
\noindent\textit{Proof.}
\paragraph{Part \ref{item: part_mu_mon}:}I show that 
\begin{align}
\label{eq: mu_comparison}
    \BP\left(\bm{\pi_\infty}(\mu_2) \geq t\right)\left\{
    \begin{array}{lr}
          \geq \BP\left(\bm{\pi_\infty}(\mu_1) \geq t\right) & \forall t\leq p \\
          \leq \BP\left(\bm{\pi_\infty}(\mu_1) \geq t\right) &  \forall t> p\,.
    \end{array}\right.
\end{align}
Observe that for every $t> p$,
\begin{equation*}
    \BP\left(\bm{\pi_\infty} \geq t\right) = m(1)+n(1) = \frac{\kappa \varphi \psi(p)/\mu}{\delta+\kappa \varphi \psi(p)/\mu}\frac{\lambda}{\delta+\lambda}\frac{(p-\alpha)\Upsilon_{2,1}(\alpha,p)}{\Upsilon_{2,1}(\alpha,p)-\alpha \Upsilon_{1,1}(\alpha,p)}\,,
\end{equation*}
which is obviously decreasing in $\mu$, hence proving the second assertion in~\eqref{eq: mu_comparison}. For every $t \leq p$:
\begin{equation*}
    \BP\left(\bm{\pi_\infty} \geq t\right)=1-\BP\left(\bm{\pi_\infty}< t\right)=1-\left(n(\alpha)+\int_\alpha^t m(\pi)\,\d \pi\right)\,.
\end{equation*}
According to~\eqref{eq: steady_state_n(a)}, the mass $n(\alpha)$ is decreasing in $\mu$, so is $\int_\alpha^t m(\pi)\,\d \pi$ by~\eqref{eq: m_CDF}. Hence, $\BP\left(\bm{\pi_\infty} \geq t\right)$ must be increasing in $\mu$ for every $t\leq p$, thus establishing the first line of~\eqref{eq: mu_comparison}. By Bayesian learning $\BE\left[\bm{\pi_\infty} (\mu_2)\right]=\BE\left[\bm{\pi_\infty} (\mu_1)\right]=p$. Therefore both parts of the Lemma~\ref{lem: suff_SSD} are satisfied to conclude Part~\ref{item: part_mu_mon}.

\paragraph{Part \ref{item: part_alpha_mon}:} Holding $\mu$ constant, for every $t\leq \alpha_2$ we have $\BP\left(\bm{\pi_\infty}(\alpha_2)\geq t\right)=1 \geq \BP\left(\bm{\pi_\infty}(\alpha_1)\geq t\right)$.
Alternatively, for every $t > \alpha_2$, it holds that
\begin{equation*}
    \BP\left(\bm{\pi_\infty}(\alpha)\geq t\right) = 1_{\{t \leq p\}}\left(\int_t^pm(\pi)\,\d \pi + n(p)\right)+n(1)+m(1)\,.
\end{equation*}
Because of~\eqref{eq: m(pi)_ind_alpha} the integral term is independent of $\alpha$ (for a fixed $\mu$). This is the case for $n(p)$ as well. Therefore, it is sufficient to show that holding $\mu$ constant, $n(1)+m(1)$ is decreasing in $\alpha$. This is equivalent to showing that the following expression is decreasing in $\alpha$:
\begin{equation*}
    \begin{split}
        \frac{(p-\alpha)\Upsilon_{2,1}(\alpha,p)}{\Upsilon_{2,1}(\alpha,p)-\alpha \Upsilon_{1,1}(\alpha,p)}&=\frac{(p-\alpha)\Upsilon_{2,1}(\alpha,p)}{(p-\alpha)p(1-p)\left(\frac{p}{\alpha}\right)^{\delta/\lambda-1}\left(\frac{1-p}{1-\alpha}\right)^{-(\delta/\lambda+2)}}\\
        &=p\left[1-\frac{\alpha^2(1-\alpha)}{p^2(1-p)}\left(\frac{\alpha}{p}\right)^{\delta/\lambda-1}\left(\frac{1-\alpha}{1-p}\right)^{-(\delta/\lambda+2)}\right]\\
        &=p\left[1-\left(\frac{\alpha}{1-\alpha}\right)^{\delta/\lambda+1}\left(\frac{p}{1-p}\right)^{-(\delta/\lambda+1)}\right]\,.
    \end{split}
\end{equation*}
Since $\alpha/(1-\alpha)$ is increasing in $\alpha$, then the above expression is decreasing in $\alpha$. So as a result of this, for every $\alpha_1\leq \alpha_2<p$ and $t>\alpha_2$, it holds that $\BP\left(\bm{\pi_\infty}(\alpha_2)\geq t\right)\leq \BP\left(\bm{\pi_\infty}(\alpha_1)\geq t\right)$. Hence, Lemma~\ref{lem: suff_SSD} can be applied again to establish Part~\ref{item: part_alpha_mon}.\qed

%----------------------------------------------------

%%%%%%%%%%%%%%%%%%%%%%%%%%%%%%%%%%%%%%%%%%%%%%%%%%%%%%%%%%%%%%%%%
\setcitestyle{numbers}	 
\bibliographystyle{normalstyle.bst}
\bibliography{ref}
%%%%%%%%%%%%%%%%%%%%%%%%%%%%%%%%%%%%%%%%%%%%%%%%%%%%%%%%%%%%%%%%%%%%%%%%%%%%%%%%%%
\end{document}